\documentclass[final,5p,times,twocolumn,numbers]{elsarticle}

\usepackage{algorithmic}
\usepackage{algorithm}
\usepackage{array}
\newcolumntype{C}{>{\centering\arraybackslash}p{1.2em}}
\usepackage{textcomp}
\usepackage{stfloats}
\usepackage[hyphens]{url}
\usepackage{multirow}
\usepackage{tikz}
\usetikzlibrary{arrows.meta,positioning,calc,fit,backgrounds,shapes.misc}
\usepackage{soul}
\usepackage{booktabs}
\usepackage{verbatim}
\usepackage{graphicx}

\usepackage{amsmath,amssymb,amsfonts,amsthm}
\allowdisplaybreaks

\usepackage{enumitem}
\usepackage{pifont}
\usepackage[caption=false,font=footnotesize,labelformat=empty]{subfig}

\usepackage{hyperref}
\hypersetup{
    colorlinks=true,
    linkcolor=blue,    
    citecolor=blue,   
    urlcolor=blue
}
\newtheorem{prop}{Proposition}
\newtheorem{theorem}{Theorem}
\newtheorem{lemma}{Lemma}
\newtheorem{assumption}{Assumption}
\newtheorem{remark}{Remark}
\def\BibTeX{{\rm B\kern-.05em{\sc i\kern-.025em b}  \kern-.08em
    T\kern-.1667em\lower.7ex\hbox{E}\kern-.125emX}}
\journal{IFAC Journal of Systems and Control
}

\begin{document}

\begin{frontmatter}

\title{A Simulation Framework with Guaranteed Error Bounds for \\ 
Safety and Fuel-Efficiency Analysis of Vehicle Platoons\tnoteref{t1}} 

\tnotetext[t1]{This work was supported by JSPS KAKENHI Grant No.~JP23K03913 and No.~JP26K07561.}
\author[sit]{Yuhao Chen}
\ead{am23102@shibaura-it.ac.jp}
\author[sit]{Ahmet Cetinkaya} 
\ead{ahmet@shibaura-it.ac.jp}

\affiliation[sit]{organization={Shibaura Institute of Technology},
            addressline={Toyosu 3-7-5}, 
            city={Koto-ku},
            postcode={135-8548}, 
            state={Tokyo},
            country={Japan}}

\begin{abstract}
Vehicle platooning is an important technology in modern transportation systems, offering significant improvements in highway traffic efficiency and fuel economy. Achieving coordinated behavior among vehicles in a platoon depends on wireless communication. However, packet losses in wireless communication can create critical safety issues when they occur together with sudden braking. In this paper, we propose a rigorous simulation-based method for studying such safety issues by analyzing the minimum inter-vehicle distance over time across control parameters that guarantee string stability. In particular, our method computes the exact distance at simulation time instants and guarantees that the change in distance between consecutive simulation time instants remains bounded. Therefore, the distances obtained at simulation times are representative of the continuous-time behavior, and the distances between those times can be accurately approximated. Our derivation relies on a lifted state representation and differential inequalities. For the proposed simulation method, we provide two approaches for selecting simulation times to ensure that the error in distance approximation remains within a given bound. We then extend our method to fuel-efficiency analysis, with guaranteed error bounds for calculating the average fuel savings of vehicles. Through an example involving a highway scenario with a merging lane, we demonstrate that among string-stable control parameter settings for a vehicle platoon, some perform better in terms of safety under simultaneous packet losses and sudden braking. We also identify control parameters that result in tradeoffs between safety and average fuel savings in a vehicle platoon.
\end{abstract}

\begin{keyword}
Vehicle Platoons, Networked Control, Packet Losses, Linear Systems, Simulation
\end{keyword}

\end{frontmatter}



\section{Introduction}
\label{sec:introduction}

A vehicle platoon is a group of vehicles that travel closely together in a coordinated formation and in the same direction. Vehicle platoons have recently become the subject of sustained technical validation efforts. For example, the European ENSEMBLE project, involving seven major European truck manufacturers, demonstrated multi-brand truck platooning through nine test campaigns conducted on test tracks and public roads \cite{willemsen2023}. The growing interest in vehicle platooning is due to recognized benefits such as enhanced traffic efficiency and reduced fuel consumption \cite{hall2005,maiti2017,sidorenko2020,xiao2022}. By maintaining small inter-vehicle distances, platoons can significantly reduce aerodynamic resistance \cite{Martinez2024} and increase highway capacity.

Despite these advantages, the transition from theoretical models to real-world highway deployment introduces important safety challenges. The behavior of vehicle platoons can be affected by dynamic environmental factors, such as road topography \cite{Gao2025} and stochastic interactions with human-driven trucks during ramp-merging maneuvers \cite{Sang2023}. The reliability of wireless communication creates additional challenges. For example, dense clustering of vehicles may overload roadside units and result in network disruptions and packet losses \cite{Pandey2025}. These communication problems become particularly important during sudden braking and other emergency maneuvers, where delayed or missing information may increase the risk of rear-end collisions.

Coordinated motion in a vehicle platoon is typically achieved through the wireless exchange of information such as acceleration and velocity. Control systems that use this shared information enable automated driving, reduce risks associated with human error \cite{hou2023,alam2014}, and are expected to improve traffic safety. The importance of rear-end collision prevention is also supported by accident statistics. According to the Japan Trucking Association \cite{JapanTruckingAssociation2024}, rear-end collisions involving large vehicles traveling straight at a constant speed accounted for 50.8\% of the accidents reported in the cited statistics for 2024. Vehicle-to-vehicle collisions also accounted for 61.9\% of highway accidents. These statistics motivate the careful evaluation of automated coordination methods for large vehicles. Traffic safety involving large vehicles has been investigated in several works (see, e.g., \cite{hall2005,liang2015,huang2019}). These studies show that  coordination can improve highway capacity, optimize platoon formation, and facilitate cooperative merging maneuvers.

A crucial aspect of vehicle platooning is string stability, which enables vehicles to maintain shorter inter-vehicle distances improving highway capacity. String stability has been extensively studied in the literature (as further discussed in Section~\ref{sec}). Nevertheless, string stability alone does not determine whether a platoon remains collision-free during a particular emergency maneuver. Under simultaneous random packet losses and sudden braking, the minimum actual inter-vehicle distance and the stopping distance depend on the transient responses of all vehicles. Reliable evaluation of these quantities requires a simulation method whose approximation accuracy is known not only at simulation time instants but also in between those time instants.

Emergency braking behavior is influenced by both the inter-vehicle spacing policy and the information available to each vehicle. Under the Predecessor-Following (PF) topology, each vehicle uses information from its immediate predecessor, whereas under the Leader-Predecessor-Following (LPF) topology, information from both the leader and the immediate predecessor is available. The spacing policy also determines the gap maintained before braking. The Constant Time Gap (CTG) policy makes the desired gap depend on vehicle speed, whereas the Constant Distance Gap (CDG) policy uses a fixed desired spacing. Therefore, both the maintained gap at the start of braking and the information source determined by the communication topology can influence the transient braking response.

Motivated by these considerations, in this paper, we develop a fast simulation approach for vehicle platoons with guaranteed approximation accuracy. Our approach builds on the continuous-time model considered in \cite{dolk2017} and substantially extends our preliminary conference results \cite{chen2025sice}. In particular, our main contributions are as follows.

First, we develop a time-discretized simulation approach that evaluates the platoon's dynamical state at discrete simulation time instants. These instants are selected so that the approximation error of the inter-vehicle distances remains within a prescribed bound at all times. Thus, the simulation intervals are determined from the required approximation accuracy rather than being fixed in advance.

Second, we introduce a faster simulation method based on a lifted state representation and the logarithmic matrix norm. This method allows simulation time gaps to be approximately ten times larger than those obtained with our first approach while maintaining the same error bound.

Third, we extend our simulation framework to fuel-efficiency analysis. By using the aerodynamic drag models from \cite{barhoumi2025} and \cite{hussein2021}, we derive guaranteed bounds on the errors that arise in the calculation of average fuel savings.

Fourth, we show that the proposed framework is not restricted to the Predecessor-Following topology or the Constant Time Gap policy. Our framework can also accommodate Leader-Predecessor-Following and $k$-Predecessor-Leader-Following topologies, as well as the Constant Distance Gap policy, by modifying the system matrices. 

Finally, we evaluate the accuracy and computational efficiency of the proposed methods by comparing them with a constant-step baseline that provides the same approximation error guarantee. We also investigate the effects of packet losses, braking parameters, the characteristic time constant, and the initial velocity on the minimum actual inter-vehicle distance and stopping distance. Furthermore, we study the relationship between the minimum actual inter-vehicle distance and fuel savings under both Constant Time Gap and Constant Distance Gap policies.

The remainder of this paper is organized as follows. In Section~\ref{sec:Relevant}, we review the relevant work and explain the research gap in detail. In Section~\ref{sec:Vehicle-Platoon-Dynamics}, we present the dynamical model of the vehicle platoon, including alternative communication topologies and spacing policies. In Section~\ref{sec:Inter-vehicle-Distances-Under}, we describe the sudden-braking and packet-loss models. In Section~\ref{sec:Time-Discretized-Simulation-Appr}, we present our proposed time-discretized simulation approach. In Section~\ref{sec:FuelEfficiency}, we extend our approach to the simulation-based calculations of fuel efficiency. In Section~\ref{sec:Numerical-Example}, we present several numerical studies. Finally, in Section~\ref{sec:Conclusion}, we conclude the paper.

\textbf{Notation:} The symbols  ${\mathbb{R}}$, ${\mathbb{N}}_{0}$, and ${\mathbb{N}}$ represent the sets of real numbers, nonnegative integers, and positive integers, respectively. The notation $\|\cdot\|$ is used for the Euclidean norm and its corresponding induced submultiplicative matrix norm. The notation $\mu(\cdot)$ denotes the logarithmic matrix norm \cite{soderlind2006logarithmic} induced by the Euclidean norm, that is, $\mu(M)\triangleq \lim_{h\to 0^+}\frac{\|I+hM\|-1}{h}$. By Theorem~24 of \cite{desoer2009feedback}, we have $\mu(M)=\lambda_{\max}((M+M^\top)/2)$ for $M\in \mathbb{R}^{n\times n}$, where $\lambda_{\max}(\cdot)$ denotes the maximum eigenvalue. For any real number $a\in{\mathbb{R}}$, $\left\lfloor a\right\rfloor$ denotes the floor function, which yields the greatest integer less than or equal to $a$. Furthermore, $W(\cdot)$ denotes the Lambert $W$ function. Important symbols used throughout the paper are listed in Table~\ref{tab}.

\begin{table}[t]
\centering
\renewcommand{\arraystretch}{1.05}
\caption{Important symbols used throughout the paper}
\label{tab}
\begin{tabular}{ll}
\toprule
\textbf{Quantity} & \textbf{Symbol} \\
\midrule
Acceleration, velocity, and position
& $a_i(t), v_i(t), p_i(t)$ \\

Desired acceleration
& $u_i(t)$ \\

Spacing error
& $e_i(t)$ \\

Actual inter-vehicle distance
& $d_i(t)$ \\

Minimum actual inter-vehicle distance
& $d'_{\min}$ \\

Time gap
& $h$ \\

Characteristic time constant
& $\tau_\mathrm{d}$ \\

Distance approximation error bound
& $\alpha$ \\

Set of communication time instants
& $\mathcal{S}$ \\

Set of simulation time instants
& $\mathcal{T}$ \\

Total amount of fuel savings
& $M_{\mathrm{save},i}$ \\

Fuel-savings approximation error
& $E_{i,k}$ \\
\bottomrule
\end{tabular}
\end{table}

\section{Relevant Work and Research Gap}
\label{sec:Relevant}

\subsection{Control, Communication Topology, and Spacing Policy}

From a control-theoretic perspective, \emph{string stability} is an important concept for realizing the efficiency benefits of vehicle platooning \cite{stankovic2000}. A platoon is string stable if disturbances, such as velocity fluctuations or changes in spacing, do not amplify as they propagate from the first vehicle to the last vehicle. To achieve string stability, \cite{dolk2017} introduced an event-triggered control approach in which each vehicle communicates its desired acceleration to the vehicle that follows it. The authors showed that an appropriate selection of the control parameters can guarantee string stability.

Although \cite{dolk2017} considers a Predecessor-Following communication topology and event-triggered communication, the continuous-time vehicle dynamics used there provide one example of a base model that can be handled by our simulation framework. In contrast to the event-triggered communication mechanism of \cite{dolk2017}, the model considered in this paper uses periodic, time-triggered communication attempts. This communication model is explained in Sections~\ref{sec:DynamicsVirtual} and~\ref{sec:CommunicationLosses}. We remark that our same simulation framework can also accommodate other communication topologies and spacing policies, as discussed in Sections~\ref{subsec:Topology} and~\ref{subsec:Policy}.

Communication topology and spacing policy also influence platoon performance. The authors of \cite{dai2025} developed a framework that unifies Predecessor-Following, Leader-Predecessor-Following, and $k$-Predecessor-Leader-Following topologies under vehicle-to-vehicle communication limitations. They reported that the $k$-Predecessor-Leader-Following topology reduced recovery time by up to 32\% compared with the Predecessor-Following topology. Motivated by such alternative architectures, we show in Section~\ref{subsec:Topology} that the proposed simulation method also applies to Leader-Predecessor-Following and $k$-Predecessor-Leader-Following topologies. Section~\ref{subsec:Policy} further shows how the framework applies to both Constant Time Gap and Constant Distance Gap spacing policies. These extensions require changes to the system matrices but do not change the main simulation methodology.

Our framework does not introduce a new platoon controller, but instead shows that the same simulation methodology can be applied to different communication topologies and spacing policies through corresponding changes in the system matrices describing the vehicle platoon control dynamics.

\subsection{Transient Communication Outages}
The effects of transient communication outages in the form of packet losses have also been studied for vehicle platoon control systems. The work in \cite{acciani2022} addressed a control problem similar to that of \cite{dolk2017} while explicitly incorporating packet losses. From a joint communication and control perspective, \cite{GIORDANO2019} designed a cooperative platoon controller that accounts for packet losses and derived bounds on the spacing error under prescribed network conditions. Unlike \cite{acciani2022} and \cite{GIORDANO2019}, which derive stability conditions and design controllers, our objective is not to develop new string-stability conditions. Instead, we evaluate the minimum actual inter-vehicle distance for given string-stable control parameters under simultaneous packet losses and sudden braking. We also provide a guaranteed bound on the simulation approximation error.

Packet losses and communication delays are critical to the performance and reliability of vehicle platoons. Their effects on platoon behavior have been investigated in several works \cite{ploeg2015,sidorenko2020,liu2022,loi2024}. The authors of \cite{Shen2024} studied mean-square stability under packet losses by using graph theory to describe the information-flow topology. The work in \cite{Wu2025} investigated whether a distributed controller could ensure mean-square stability under Markovian packet losses. The authors of \cite{Zhao2025} developed methods for maintaining the desired spacing between vehicles under communication delays and high-speed road conditions. Vehicle platoons communicating through general ad-hoc networks were studied in \cite{xiao2022}. Furthermore, \cite{Santini2019} evaluated distributed platoon control during joining and leaving maneuvers by using realistic vehicle dynamics and IEEE 802.11p communication models.

In our model, a packet loss represents a transient communication outage at an individual periodic communication instant. It does not represent a complete or permanent loss of the communication system. A large packet-loss probability or a long sequence of consecutive packet losses represents more severe communication degradation. The simplified packet-loss model allows us to isolate the effects of missing information on the vehicle dynamics. More detailed packet-loss sequences generated by network simulators can also be incorporated into the proposed framework. Thus, the contribution of the proposed method in this setting is the simulation of vehicle dynamics for prescribed packet-loss sequences with a guaranteed approximation error bound.

\subsection{Emergency Braking}

Emergency braking has been studied from several perspectives. The works in \cite{hu2021} and \cite{Sidorenko2022} derived minimum-distance models for emergency braking under communication delays. The authors of \cite{segata2015communication} evaluated communication strategies for vehicle platoons by considering deviations in inter-vehicle distances as a safety-related metric. In the present paper, we evaluate the minimum actual inter-vehicle distance over the complete braking maneuver rather than only the distance after all vehicles have stopped. In the adopted longitudinal model, a zero minimum actual inter-vehicle distance corresponds to a collision. Our framework provides guarantees on the numerical accuracy of the computed distance approximation rather than a general safety guarantee for the vehicle platoon.

\subsection{Simulation Frameworks and Fuel-Efficiency Analysis}

Discrete-event and co-simulation platforms such as Veins and PLEXE \cite{segata2014PLEXE,segata2022multi,Hasan2023} integrate communication networks, traffic mobility, and vehicle dynamics. These platforms are important for studying realistic physical-layer and medium-access-control behavior. The objective of our approach is different. We address the numerical accuracy of the continuous-time vehicle-dynamics calculation. In particular, we select simulation intervals so that the approximation error of the inter-vehicle distances remains within a prescribed bound between consecutive simulation time instants. Therefore, the proposed method is complementary to existing network and traffic simulators. Packet-loss sequences obtained from such simulators can be used as inputs to our framework. Similarly, the proposed vehicle-dynamics calculation may be incorporated into a broader co-simulation environment.

General-purpose numerical methods, such as Runge-Kutta methods, are commonly used to solve initial-value problems \cite{griffiths2010numerical}. Both fixed-step and variable-step numerical solvers approximate system trajectories according to selected step sizes or solver tolerances. Our proposed approach is also a variable-step method, but its objective is different. It selects each simulation interval directly from a prescribed bound on the inter-vehicle distance approximation. The selected interval depends on the current system state and input. Consequently, the interval can become smaller during strong transient behavior, such as rapid deceleration, and larger during mild changes in the vehicle states. Section~\ref{sec:Numerical-Example} compares our approaches with a constant-step baseline that satisfies the same approximation error requirement.

Fuel efficiency is another important aspect of platoon operation. The aerodynamic effects of inter-vehicle spacing and the resulting fuel savings have been investigated in \cite{barhoumi2025,Kim2025APJ}. These studies characterize the fuel savings achieved by platooning. However, they do not provide guaranteed bounds on the numerical approximation error arising when fuel savings are calculated from values available only at discrete simulation time instants. In this paper, we use the aerodynamic drag models in \cite{barhoumi2025} and \cite{hussein2021} and derive explicit bounds on the resulting fuel-savings approximation error. In this respect, our proposed framework links the simulation interval directly to the required approximation accuracy and extends the resulting guarantee to fuel-savings calculations.

The research gap addressed in this paper concerns the combined evaluation of packet losses, emergency braking, communication topology, and fuel efficiency through a computationally efficient simulation method that provides a guaranteed bound on the inter-vehicle distance approximation throughout every simulation interval. To the best of our knowledge, the works reviewed above do not provide such a bound under the combined effects of random packet losses and sudden braking. They also do not extend this type of guarantee to the calculation of fuel savings.

A combined evaluation is important because random packet losses require many independent simulation runs, which can result in high computational cost. At the same time, reducing the number of simulation instants without controlling the behavior between those instants may lead to an inaccurate value of the minimum actual inter-vehicle distance. Our approach addresses both concerns by selecting the simulation intervals according to the required approximation accuracy. The numerical studies illustrate the proposed approach by examining the effects of packet losses, control parameters, braking parameters, initial velocity, communication topology, and spacing policy on the minimum actual inter-vehicle distance, stopping distance, and fuel savings.

Furthermore, we note that this paper extends our preliminary conference work \cite{chen2025sice}, in multiple ways. First, we introduce a new lifted-state approach that permits considerably larger simulation intervals while maintaining the same error bound. Second, we provide a complete explicit characterization of the braking model, which enables direct comparison with the method in \cite{chen2025sice}. Third, we extend the framework to fuel-efficiency calculations and derive guaranteed bounds on the fuel-savings approximation error. Fourth, we show how the method applies to Leader-Predecessor-Following and $k$-Predecessor-Leader-Following topologies and to the Constant Distance Gap policy. Finally, we add comparisons with a constant-step baseline and conduct expanded sensitivity studies involving packet-loss probability, braking parameters, characteristic time constant, initial velocity, minimum actual inter-vehicle distance, and stopping distance. The implementation and scripts used for the numerical studies are made publicly available in the accompanying GitHub repository\footnote{\url{https://github.com/chenyuhao0001/safepsim}} to support reproducibility.

\begin{figure}[t]
\centering \includegraphics[width=1\columnwidth]{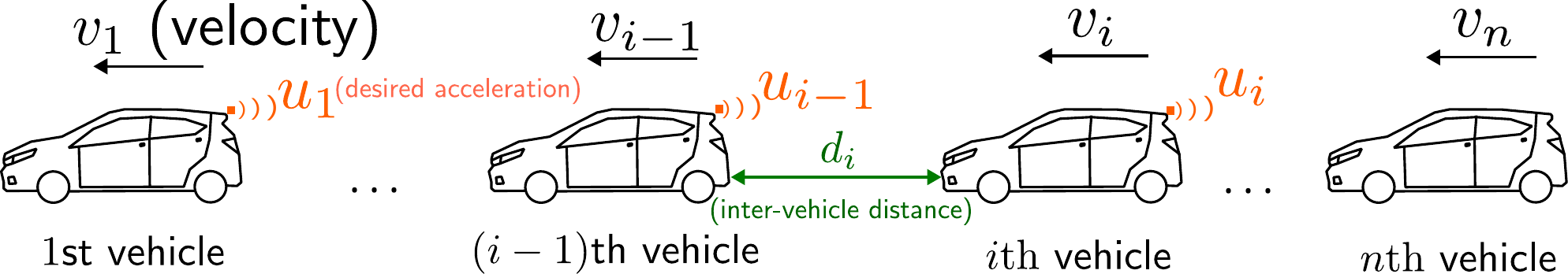}
\caption{Illustration of a vehicle platoon with $n$ vehicles.}
 \label{Figure-vehicle-platoon}
\end{figure}

\section{Vehicle Platoon Dynamics}
\label{sec:Vehicle-Platoon-Dynamics}

In this paper, we consider a vehicle platoon model similar to that
proposed by \cite{dolk2017}. This model incorporates a virtual reference
vehicle and $n\in{\mathbb{N}}$ actual vehicles in the platoon. Our
model differs from that of \cite{dolk2017} in that it incorporates
packet losses and a simultaneously occurring sudden brake, as these issues were not considered in the work of \cite{dolk2017}.
Figure~\ref{Figure-vehicle-platoon} shows an illustration of the
platoon setting that we consider. This model uses the so-called Predecessor-Following communication topology (see top lane in Figure~\ref{Figure-topology}), where each vehicle uses wireless information only from its immediate predecessor, as well as the Constant Time Gap spacing policy, under which each vehicle aims to keep a desired distance that increases linearly with its own velocity. We will explain how other topologies and spacing policies can be modeled later in Sections~\ref{subsec:Topology} and~\ref{subsec:Policy}.

\begin{figure}
    \centering \includegraphics[width=1\columnwidth]{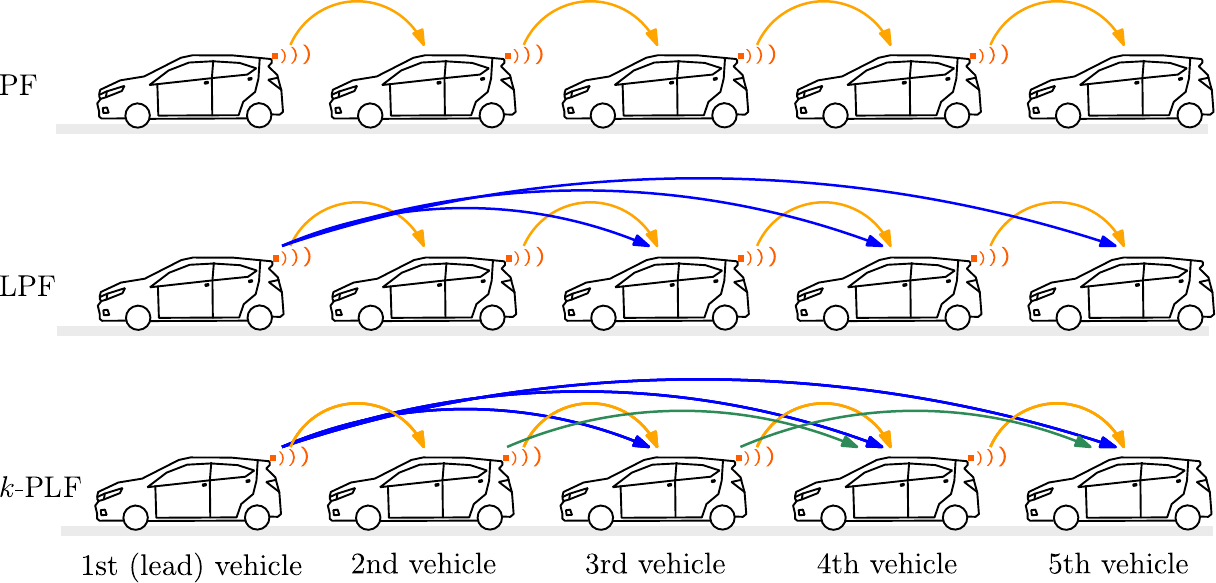}
\caption{{Illustration of wireless information flow in different vehicle platoon communication topologies. Top: Predecessor-Following topology (PF), Middle: Leader-Predecessor-Following topology (LPF), Bottom: $k$-Predecessor-Leader-Following topology ($k$-PLF) with $k=2$.}}
 \label{Figure-topology}
\end{figure}

\subsection{Dynamics of virtual reference vehicle and actual vehicles}
\label{sec:DynamicsVirtual}

In the vehicle platoon, the virtual reference vehicle's motion is
described by
\begin{align}
\dot{p}_{0}(t) & =v_{0}(t),\label{eq:p0-def}\\
\dot{v}_{0}(t) & =a_{0}(t),\label{eq:v0-def}\\
\dot{a}_{0}(t) & =-\frac{1}{\tau_{{\mathrm{d}}}}a_{0}(t)+\frac{1}{\tau_{{\mathrm{d}}}}u_{0}(t),\quad t\geq0,\label{eq:a0-def}
\end{align}
where $p_{0}(t)\in{\mathbb{R}}$ and $v_{0}(t)\in{\mathbb{R}}$ respectively
denote the position and the velocity at time $t$; $a_{0}(t)\in{\mathbb{R}}$
and $u_{0}(t)\in{\mathbb{R}}$ respectively denote the acceleration
and the desired acceleration at time $t$; moreover, $\tau_{{\mathrm{d}}}>0$
is the so-called characteristic time constant that determines how
fast the vehicle can react. The $1$st actual vehicle follows the
virtual reference vehicle. 

For each $i\in\{1,2,\ldots,n\}$, the motion of the $i$th actual
vehicle is described by 
\begin{align}
\dot{e}_{i}(t) & =v_{i-1}(t)-v_{i}(t)-ha_{i}(t),\label{eq:e-def}\\
\dot{p}_{i}(t) & =v_{i}(t),\label{eq:pi-def}\\
\dot{v}_{i}(t) & =a_{i}(t), \label{eq:vi-def-first}\\
\dot{a}_{i}(t) & =-\frac{1}{\tau_{{\mathrm{d}}}}a_{i}(t)+\frac{1}{\tau_{{\mathrm{d}}}}u_{i}(t),\\
\dot{u}_{i}(t) & =-\frac{1}{h}u_{i}(t)+\frac{1}{h}\left(k_{{\mathrm{p}}}e_{i}(t)+k_{{\mathrm{d}}}\dot{e}_{i}(t)+\hat{u}_{i-1}(t)\right).\label{eq:ui-def}
\end{align}
In (\ref{eq:e-def}), the scalar $e_{i}(t)$ denotes the spacing error
defined through the \emph{constant time-gap policy} as 
\begin{align*}
e_{i}(t) & \triangleq(p_{i-1}(t)-p_{i}(t)-L_{i})-(r+hv_{i}(t)),
\end{align*}
where $L_{i}>0$ denotes the length of vehicle $i$, $r>0$ denotes
the standstill distance, and $h$ denotes the time gap. Notice that
the term $p_{i-1}(t)-p_{i}(t)-L_{i}$ represents the actual inter-vehicle
distance and $r+hv_{i}(t)$ represents the desired distance. In (\ref{eq:pi-def})--(\ref{eq:ui-def}),
$p_{i}(t),v_{i}(t),a_{i}(t),u_{i}(t)\in{\mathbb{R}}$ respectively denote
the position, the velocity, the acceleration, and the desired acceleration
of the $i$th vehicle at time $t$; moreover, $k_{{\mathrm{p}}},k_{{\mathrm{d}}}\in{\mathbb{R}}$
are control parameters, and $\hat{u}_{i-1}(t)\in{\mathbb{R}}$ denotes
the desired acceleration of the $(i-1)$th vehicle that is most recently
received by the $i$th vehicle over wireless channel. 

We note that the time gap constant $h$  in \eqref{eq:ui-def} affects the speed of convergence of $u_i$ towards $k_{{\mathrm{p}}}e_{i}(t)+k_{{\mathrm{d}}}\dot{e}_{i}(t)+\hat{u}_{i-1}(t)$. It is possible to see this by treating $k_{{\mathrm{p}}}e_{i}(t)+k_{{\mathrm{d}}}\dot{e}_{i}(t)+\hat{u}_{i-1}(t)$ as a constant target value denoted as $c$. Then \eqref{eq:ui-def} yields $u_i(t) = c - e^{-\frac{t}{h}}(c-u_i(0))$, indicating that $u_i(t)$ converges to $c$ with $1/h$ as the speed of convergence. 

For sufficiently frequent communications between vehicles, \cite{dolk2017}
showed that there are $k_{{\mathrm{p}}}$ and $k_{{\mathrm{d}}}$ values
that guarantee string stability, which ensures that speed changes
and fluctuations in distance between vehicles do not amplify as they
propagate through the platoon. In this paper our goal is not to design
$k_{{\mathrm{p}}}$ and $k_{{\mathrm{d}}}$, but to propose a simulation framework that can be used to accurately compute inter-vehicle distances in a platoon when the communication is subject to packet losses and there is a sudden brake.

We assume that each actual vehicle attempts to communicate with the
next vehicle periodically at times $s_{0},s_{1},s_{2},\ldots\in[0,\infty)$,
where $s_{0}=0$ and $s_{j+1}-s_{j}=T$ with $T>0$ denoting the communication
interval. We collect the communication time instants in the set ${\mathcal{S}}\triangleq\{s_{0},s_{1},s_{2},\ldots\}$.
Furthermore, we use $l_{0}^{i},l_{1}^{i},l_{2}^{i},\ldots\in\{0,1\}$
to indicate the packet losses at those times. Specifically,
$l_{j}^{i}=1$ indicates that the communication attempt of the $i$th
vehicle at time $s_{j}$ faces a packet loss and $l_{j}^{i}=0$ indicates
that it is successful. Assuming that the initial communication
attempts at time $s_{0}$ are successful, for each $i\in\{1,2,\ldots,n-1\}$,
$\hat{u}_{i}(t)$ can be characterized as 
\begin{align}
\hat{u}_{i}(s_{0}) & =u_{i}(s_{0}),\label{eq:hatu-begin}\\
\hat{u}_{i}(s_{j}) & =\begin{cases}
\hat{u}_{i}(s_{j-1}), & {\mathrm{if}}\,\,l_{j}^{i}=1\,\,({\mathrm{packet\,\,loss}}),\\
u_{i}(s_{j}), & {\mathrm{if}}\,\,l_{j}^{i}=0\,\,({\mathrm{success}}),
\end{cases}\,\,j\in{\mathbb{N}},\label{eq:hat-u-mid}\\
\hat{u}_{i}(t) & =\hat{u}_{i}(s_{j}),\quad t\in(s_{j},s_{j+1}),\quad j\in{\mathbb{N}}_{0}.\label{eq:ui-piecewise-constant}
\end{align}
Note that for $i=1$, the equation (\ref{eq:ui-def}) includes
the term $\hat{u}_{0}(t)$. Since this is the virtual reference vehicle's
desired acceleration, it does not require communication to be known
by the $1$st vehicle, and therefore, we set
\begin{align}
\hat{u}_{0}(t) & =u_{0}(t),\quad t\geq0.\label{eq:hatu-end}
\end{align}
\begin{remark} The equality in (\ref{eq:ui-piecewise-constant})
implies that for $i\in\{1,2,\ldots,n-1\}$, the function $\hat{u}_{i}\colon[0,\infty)\to{\mathbb{R}}$
is a piecewise-constant function of time. This point is utilized in
Section~\ref{sec:Time-Discretized-Simulation-Appr} to propose a
discrete-time approach for simulating the dynamics of the vehicle
platoon. \end{remark}
\subsection{The overall dynamics as a linear system}
\label{sec:TheOverallDynamics}
The overall dynamics of the vehicle platoon can be described as a
linear system of the form
\begin{align}
\dot{x}(t) & =A_{{\mathrm{c}}}x(t)+B_{{\mathrm{c}}}u(t).\label{eq:linear-system}
\end{align}
Here, the state vector $x(t)\in{\mathbb{R}}^{3+6n}$ is given by 
\begin{align}
x(t) & \triangleq\left[\begin{array}{ccccc}
x_{0}^{\top}(t) & x_{1}^{\top}(t) & x_{2}^{\top}(t) & \cdots & x_{n}^{\top}(t)\end{array}\right]^{\top},\label{eq:x-def}
\end{align}
where $x_{0}(t)\triangleq[p_{0}(t),v_{0}(t),a_{0}(t)]^{\top}$ and
\begin{align*}
x_{i}(t) & \triangleq\left[e_{i}(t),\dot{e}_{i}(t),p_{i}(t),v_{i}(t),a_{i}(t),u_{i}(t)\right]^{\top},
\end{align*}
for $i\in\{1,2,\ldots,n\}$. Moreover, the input vector $u(t)\in{\mathbb{R}}^{1+n}$
is given by 
\begin{align}
u(t) & \triangleq\left[u_{0}(t),\hat{u}_{0}(t),\hat{u}_{1}(t),\ldots,\hat{u}_{n-1}(t)\right]^{\top},\label{eq:u-def}
\end{align}
 where $u_{0}(t)$ is the desired acceleration for the virtual reference
vehicle and $\hat{u}_{i}(t)$, $i\in\{0,1,\ldots,n-1\}$, are given
by (\ref{eq:hatu-begin})--(\ref{eq:hatu-end}).

By (\ref{eq:p0-def})--(\ref{eq:ui-def}), the matrices $A_{{\mathrm{c}}}\in{\mathbb{R}}^{(3+6n)\times(3+6n)}$,
$B_{{\mathrm{c}}}\in{\mathbb{R}}^{(3+6n)\times(1+n)}$ can be given in
block-matrix form as 
\begin{align*}
A_{{\mathrm{c}}} & \triangleq\left[\begin{array}{@{}ccccc@{}}
M & 0 & 0 & \cdots & 0\\
G & N & 0 & \ddots & \vdots\\
0 & H & N & \ddots & 0\\
\vdots & \ddots & \ddots & \ddots & 0\\
0 & \cdots & 0 & H & N
\end{array}\right],\,\,
B_{{\mathrm{c}}} \triangleq\left[\begin{array}{@{}cccc@{}}
E & 0 & \cdots & 0\\
0 & F & \ddots & \vdots\\
\vdots & \ddots & \ddots & 0\\
0 & \cdots & 0 & F
\end{array}\right],
\end{align*}
where 
\begin{align*}
M & \triangleq\left[\begin{array}{ccc}
0 & 1 & 0\\
0 & 0 & 1\\
0 & 0 & -1/\tau_{{\mathrm{d}}}
\end{array}\right],\\
N & \triangleq\left[\begin{array}{cccccc}
0 & 0 & 0 & -1 & -h & 0\\
0 & 0 & 0 & 0 & h/\tau_{{\mathrm{d}}}-1 & -h/\tau_{{\mathrm{d}}}\\
0 & 0 & 0 & 1 & 0 & 0\\
0 & 0 & 0 & 0 & 1 & 0\\
0 & 0 & 0 & 0 & -1/\tau_{{\mathrm{d}}} & 1/\tau_{{\mathrm{d}}}\\
k_{{\mathrm{p}}}/h & k_{{\mathrm{d}}}/h & 0 & 0 & 0 & -1/h
\end{array}\right],\\
G & \triangleq\left[\begin{array}{ccc}
0 & 1 & 0\\
0 & 0 & 1\\
0 & 0 & 0\\
0 & 0 & 0\\
0 & 0 & 0\\
0 & 0 & 0
\end{array}\right],\,\,H\triangleq\left[\begin{array}{cccccc}
0 & 0 & 0 & 1 & 0 & 0\\
0 & 0 & 0 & 0 & 1 & 0\\
0 & 0 & 0 & 0 & 0 & 0\\
0 & 0 & 0 & 0 & 0 & 0\\
0 & 0 & 0 & 0 & 0 & 0\\
0 & 0 & 0 & 0 & 0 & 0
\end{array}\right],\\
E & \triangleq\left[\begin{array}{c}
0\\
0\\
1/\tau_{{\mathrm{d}}}
\end{array}\right],\,\,F\triangleq\left[\begin{array}{c}
0\\
0\\
0\\
0\\
0\\
1/h
\end{array}\right].
\end{align*}
In Section~\ref{sec:Time-Discretized-Simulation-Appr}, the continuous-time
linear dynamics in (\ref{eq:linear-system}) will  provide the basis
for a discrete-time approach for simulating the dynamics of the vehicle
platoon.

\begin{remark}
The linear vehicle platoon model adopted here is consistent with a well-established line of work in cooperative vehicle control \cite{dolk2017, acciani2022, stankovic2000} and is appropriate for the highway cruising regime that we consider. In highway conditions at near constant speed, the dominant nonlinearities in longitudinal vehicle dynamics arise from aerodynamic drag and powertrain saturation. The former is addressed separately in our fuel efficiency analysis in Section~\ref{sec:FuelEfficiency}, and the latter is implicitly captured through the deceleration parameter $\gamma$ in the braking model that we will explain in Section~\ref{sec:breakingmodel}. The linear structure enables closed form expressions for inter-vehicle distances via matrix exponentials, which allows us to derive guaranteed error bounds in Section~\ref{sec:Time-Discretized-Simulation-Appr}. 
\end{remark}

In the remainder of this section, we show that the platoon model that we consider can further be used to describe different vehicle-to-vehicle communication network topologies and different spacing policies.

\subsubsection{Modeling Leader-Predecessor-Following (LPF) and $k$-Predecessor-Leader-Following Topologies ($k$-PLF)}
\label{subsec:Topology}
In the Predecessor-Following (PF) topology adopted above, each vehicle $i$ receives the desired acceleration $\hat{u}_{i-1}(t)$ only from its immediate predecessor. Our framework can accommodate other communication topologies. As a first example, we consider the Leader-Predecessor-Following (LPF) topology, in which each vehicle additionally receives the desired acceleration $\hat{u}_1(t)$ of the leader vehicle (see middle lane in Figure~\ref{Figure-topology}). Following the LPF structure discussed in \cite{dai2025}, the control input in \eqref{eq:ui-def} is given by  
\begin{align}
\dot{u}_{i}(t) & =-\frac{1}{h}u_{i}(t)+\frac{1}{h}\left(k_{{\mathrm{p}}}e_{i}(t)+k_{{\mathrm{d}}}\dot{e}_{i}(t)+w_\mathrm{P}\hat{u}_{i-1}(t) + w_\mathrm{L}\hat{u}_1(t)\right), \label{eq:ui-lpf}
\end{align}
where $w_\mathrm{P}>0$ represents the weight coefficient of the predecessor vehicle, $w_\mathrm{L}>0$ represents the weight coefficient of the lead vehicle, satisfying $w_\mathrm{P} + w_\mathrm{L} =1$. Since $e_i(t)$, $\dot{e}_i(t)$, and $u_i(t)$ remain the state variables of vehicle $i$, the matrix $A_\mathrm{c}$ is the same as shown previously. The platoon dynamics are given again by \eqref{eq:linear-system}, where $B_{\mathrm c}$ is replaced with a new matrix $B^{\mathrm{LPF}}_{{\mathrm{c}}}\in{\mathbb{R}}^{(3+6n)\times(1+n)}$ that is given in block form as
\begin{align*}
B^{\mathrm{LPF}}_{{\mathrm{c}}} & \triangleq\left[\begin{array}{@{}cccccc@{}}
E & 0 & 0 & 0 & \cdots & 0\\
0 & F & 0 & 0 & \cdots & 0\\
0 & 0 & F & 0 & \cdots & 0\\
0 & 0 & F_\mathrm{L} & F_\mathrm{P} & \ddots & \vdots\\
\vdots & \vdots & \vdots & \ddots & \ddots & 0\\
0 & 0 & F_\mathrm{L} & 0 & \cdots & F_\mathrm{P}
\end{array}\right],
\end{align*}
where $F_\mathrm{L}\triangleq w_{\mathrm{L}}F$ and $F_\mathrm{P}\triangleq w_{\mathrm{P}}F.$

The same approach also applies to the $k$-Predecessor-Leader-Following ($k$-PLF) topology (see bottom lane in Figure~\ref{Figure-topology}), in which
\begin{align}
\dot{u}_{i}(t) & =-\frac{1}{h}u_{i}(t)+\frac{1}{h}\Bigg(k_{{\mathrm{p}}}e_{i}(t)+k_{{\mathrm{d}}}\dot{e}_{i}(t) \nonumber \\ 
& \qquad +\sum_{m=1}^{k}w_{\mathrm{P}_m}\hat{u}_{i-m}(t) + w_\mathrm{L}\hat{u}_1(t)\Bigg), \label{eq:ui-klpf}
\end{align}
where $w_{\mathrm{P}_m}>0$ represents the weight coefficient of the $m$th predecessor vehicle for $m\in\{1,\ldots,k\}$, $w_\mathrm{L}>0$ represents the weight coefficient of the lead vehicle, satisfying $\sum_{m=1}^{k}w_{\mathrm{P}_m} + w_\mathrm{L} =1$. In this topology, each vehicle $i$ receives the desired accelerations of its $k$ nearest predecessors $\hat{u}_{i-1}(t),\ldots,\hat{u}_{i-k}(t)$ together with that of the lead vehicle $\hat{u}_1(t)$. This reduces to the LPF when $k=1$. Since $e_i(t)$, $\dot{e}_i(t)$, and $u_i(t)$ remain the state variables of vehicle $i$, the matrix $A_\mathrm{c}$ is the same as shown previously. The platoon dynamics are given again by \eqref{eq:linear-system}, with $B_{\mathrm c}$ replaced with a new matrix $B^{k\text{-}\mathrm{PLF}}_{{\mathrm{c}}}\in{\mathbb{R}}^{(3+6n)\times(1+n)}$ that is given in block form as
\begin{align*}
B^{k\text{-}\mathrm{PLF}}_{{\mathrm{c}}} & \triangleq\left[\begin{array}{@{}cccccccc@{}}
E & 0 & 0 & 0 & 0 & \cdots & 0 & 0\\
0 & F & 0 & 0 & 0 & \cdots & 0 & 0\\
0 & 0 & F_\mathrm{L} & 0 & 0 & \cdots & 0 & 0\\
0 & 0 & F_\mathrm{L} & F_{\mathrm{P}_1} & 0 & \cdots & 0 & 0\\
0 & 0 & F_\mathrm{L} & F_{\mathrm{P}_2} & F_{\mathrm{P}_1} & \ddots & 0 & 0\\
\vdots & \vdots & \vdots & \vdots & \ddots & \ddots & & \vdots\\
0 & 0 & F_\mathrm{L} & F_{\mathrm{P}_k} & F_{\mathrm{P}_{k-1}} & \cdots & F_{\mathrm{P}_1} & 0\\
\vdots & \vdots & \vdots & & \ddots & \ddots & \ddots & \vdots\\
0 & 0 & F_\mathrm{L} & 0 & \cdots & F_{\mathrm{P}_k} & \cdots & F_{\mathrm{P}_1}
\end{array}\right],
\end{align*}
where $F_\mathrm{L}\triangleq w_\mathrm{L}F$ and $F_{\mathrm{P}_m}\triangleq w_{\mathrm{P}_m}F$ for $m\in\{1,\ldots,k\}$.

\subsubsection{Modeling Constant Distance Gap (CDG) Spacing Policy}
\label{subsec:Policy}
As a further example, we consider the Constant Distance Gap (CDG) policy, following the formulation in \cite{besselink2017}. In CDG policy, the spacing error in \eqref{eq:e-def} is redefined as
\begin{align}
e_i^{\mathrm{CDG}}(t) \triangleq (p_{i-1}(t)-p_i(t)-L_i) - d_{\mathrm{des}}, \label{eq:e-cdg}
\end{align}
where $d_{\mathrm{des}}>0$ is the desired distance between vehicles. As a consequence, we have $\dot{e}_i^{\mathrm{CDG}}(t) = v_{i-1}(t)-v_i(t)$. This implies that the platoon dynamics remain of the linear form of \eqref{eq:linear-system}, with $A_\mathrm{c}$ replaced with a new matrix $A_\mathrm{c}^{\mathrm{CDG}}$. This new matrix $A_\mathrm{c}^{\mathrm{CDG}}$ that has the same block matrix form of $A_\mathrm{c}$, where $N$ is replaced with a new matrix  $N^{\mathrm{CDG}}$ given by 
\begin{align}
N^{\mathrm{CDG}} \triangleq \left[\begin{array}{cccccc}
0 & 0 & 0 & -1 & 0 & 0\\
0 & 0 & 0 & 0 & -1 & 0\\
0 & 0 & 0 & 1 & 0 & 0\\
0 & 0 & 0 & 0 & 1 & 0\\
0 & 0 & 0 & 0 & -1/\tau_\mathrm{d} & 1/\tau_\mathrm{d}\\
k_\mathrm{p}/h & k_\mathrm{d}/h & 0 & 0 & 0 & -1/h
\end{array}\right]. \label{eq:N-cdg}
\end{align}
The discussion above shows that the characterization in \eqref{eq:linear-system} allows a variety of different vehicle platoon settings by changing $A_{\mathrm{c}}$ and $B_{\mathrm{c}}$ matrices. A consequence is that the simulation approach that we will develop in Section~\ref{sec:Time-Discretized-Simulation-Appr} can handle all of these different platoon settings (PF, LPF, $k$-PLF, CTG, CDG), since we will use \eqref{eq:linear-system} in deriving our approach.

\section{Inter-vehicle Distances Under Sudden Braking and Packet Losses}
\label{sec:Inter-vehicle-Distances-Under}

Our goal in this paper is to compute how distances between consecutive
vehicles change when the $1$st vehicle brakes suddenly. While string
stability ensures that sudden changes in the velocity do not propagate
in the platoon, a sudden brake may cause collisions in a platoon if
it happens in conjunction with packet losses in the wireless
channel used by the vehicles. When the leading vehicle brakes suddenly under packet losses, the following vehicles may not receive the warning in time. This  can lead to insufficient braking time, and in the worst case it may lead to vehicle collisions. 

\subsection{Inter-vehicle distances}

We can identify collisions through inter-vehicle distances. The inter-vehicle
distance in front of the $i$th vehicle is defined as 
\begin{align}
d_{i}(t) & \triangleq(p_{i-1}(t)-p_{i}(t)-L_{i}),\quad i\in\{2,\ldots,n\}.\label{eq:di-def}
\end{align}
We emphasize that $d_i(t)$ in (\ref{eq:di-def}) is the actual physical inter-vehicle distance at the time instant $t$. Notice that this definition accounts for the length of each vehicle.
We will use the solutions $x(t)$ of (\ref{eq:linear-system}) to
compute the inter-vehicle distances $d_{i}(t)$ as 
\begin{align}
d_{i}(t) & =q_{i}^{\top}x(t)-L_{i},\label{eq:di-equality}
\end{align}
where $q_{i}\in{\mathbb{R}}^{3+6n}$ is a column vector with $j$th entry given as 
\begin{align}
(q_i)_j & \triangleq\begin{cases}
-1, & {\mathrm{if}}\,\,j=6i,\\
1, & {\mathrm{if}}\,\,j=6i-6,\\
0, & {\mathrm{otherwise}},
\end{cases}\quad j\in\{1,2,\ldots,3+6n\}.\label{eq:vij-def}
\end{align}
 Notice that (\ref{eq:di-equality}) holds, since $q_{i}^{\top}x(t)=p_{i-1}(t)-p_{i}(t)$. 

To identify if any vehicle in the platoon gets close to a collision
in a given time interval $[0,t_{{\mathrm{end}}}]$ with $t_{\mathrm{end}}>0$, we also define the
minimum of minimum inter-vehicle distances as 
\begin{align}
d_{\min}^{*} & =\min_{i\in\{2,\ldots n\}}\min_{t\in[0,t_{{\mathrm{end}}}]}d_{i}(t).\label{eq:d-min-def}
\end{align}
Notice that the minimum is taken over vehicles across all time instants. Therefore, $d_{\min}^{*}\leq0$ indicates that a collision has occurred
somewhere in the platoon at some time between $t=0$ and $t=t_{{\mathrm{end}}}$. If for instance, a braking maneuver happens in the interval, then the minimum actual inter-vehicle distance is calculated over the complete braking maneuver, not only at the end of the maneuver. This point is important because a collision may occur during the transient response even when the final standstill distances are positive.  
In Section~\ref{sec:Time-Discretized-Simulation-Appr}, we will provide
a method to approximate $d_{\min}^{*}$ reliably through simulation. Before
we discuss this method, we first explain the braking
and the packet loss models that we consider.

\subsection{Braking model}

\label{sec:breakingmodel}

We model the braking by modifying the virtual reference vehicle's
desired acceleration $u_{0}(t)$. Notice that the virtual reference
vehicle's dynamics are used directly as a reference by the $1$st
vehicle. In particular, the information of $\hat{u}_{0}(t)=u_{0}(t)$
is directly available to the $1$st vehicle. 

The braking model that we consider is of the form 
\begin{equation} \label{eq:brake-model}
u_{0}(t) = \begin{cases}
0, & t<t_{\mathrm{brake}},\\
\max\{-\gamma,-\eta v_{0}(t)\}, & t\geq t_{\mathrm{brake}}.
\end{cases}
\end{equation}
 This braking model involves three parameters: the initial braking time $t_{\mathrm{brake}}>0$ and the deceleration parameters $\gamma>0$ and $\eta>0$. Since we consider forward motion, we have $v_{0}(t)\geq0$ for all $t\geq0$. Consequently, the desired acceleration satisfies $u_{0}(t)\leq0$ for $t\geq t_{\mathrm{brake}}$. The braking model's behavior changes based on $v_{0}(t)$. When $v_{0}(t)$ is sufficiently small such that $v_{0}(t) < \gamma/\eta$, the expression in \eqref{eq:brake-model} can be simplified to $\max\{-\gamma,-\eta v_{0}(t)\}=-\eta v_{0}(t)$. Hence, as the virtual reference vehicle's velocity approaches $0$, the desired acceleration $u_{0}(t)$ also approaches $0$. 

To further analyze the behavior resulting from the braking model in \eqref{eq:brake-model}, we provide a few technical results. Let $t^*>0$ be the time instant at which $u_0(t)$ changes from $-\gamma$ to $-\eta v_0(t)$ and since $v_0(t)$ decreases monotonically after $t> t_\mathrm{brake}$, a larger value of $\eta$ corresponds to a smaller value of $v_0(t^*)$ and larger value of $t^*$. This means that the constant braking deceleration value $\gamma$ is held longer, and the velocity converges to zero faster. Specifically,
\begin{align}
t^* \triangleq \min\{t \geq t_\mathrm{brake}:-\eta v_0(t) \geq -\gamma \}. \label{eq:tstar definiton}
\end{align}
The following result characterizes the value of $t^*$. 
 \begin{lemma} \label{Lemma-tstar}
  Consider the time constant $t^*$ defined in (\ref{eq:tstar definiton}) in relation to the braking model in \eqref{eq:brake-model}. If $v_0(t_\mathrm{brake}) > \gamma/\eta$, then
 \begin{align}
 t^* &= -\tau_\mathrm{d}\beta_1 + \tau_\mathrm{d}W \left(\frac{-(a_0 (t_\mathrm{brake})+\gamma)e^{\frac{t_\mathrm{brake}}{\tau_\mathrm{d}}+\beta_1}}{\gamma}\right), \label{eq:tstar calculation}
 \end{align}
 where $W$ denotes the Lambert $W$ function and 
 \begin{align*}
 \beta_1 \triangleq {\frac{\gamma/\eta-v_0(t_\mathrm{brake})-\tau_\mathrm{d}a_0(t_\mathrm{brake})-\gamma\tau_\mathrm{d}-\gamma t_\mathrm{brake}}{\gamma\tau_\mathrm{d}}}.
 \end{align*}
 \end{lemma} \medskip
 \begin{proof}
First of all, the condition $v_0(t_\mathrm{brake}) > \gamma/\eta$ implies that $t^* > t_\mathrm{brake}$. Therefore, for $t\in[t_\mathrm{brake} ,t^*]$, (\ref{eq:a0-def}) becomes
\begin{align}
\dot{a}_0(t) = -\frac{1}{\tau_\mathrm{d}} a_0(t) -\frac{\gamma}{\tau_\mathrm{d}}, \label{eq:a0-def2}
\end{align}
which implies that
\begin{align}
a_0(t) = (a_0(t_\mathrm{brake})+\gamma)e^{\frac{t_\mathrm{brake}-t}{\tau_\mathrm{d}}}-\gamma. \label{eq:a0-def2intsim}
\end{align}
Furthermore, using (\ref{eq:v0-def}), we obtain that
\begin{align}
v_0(t) = v_0(t_\mathrm{brake}) + \int_{t_\mathrm{brake}}^{t}a_0(\tau)\mathrm{d}\tau. \label{eq:v0-def2}
\end{align}
By substituting (\ref{eq:a0-def2intsim}) into (\ref{eq:v0-def2}),  for  $t\in[t_\mathrm{brake} ,t^*]$, we get
\begin{align}
v_0(t) &= v_0 (t_\mathrm{brake}) + \tau_\mathrm{d}(a_0(t_\mathrm{brake})+\gamma)(1-e^{\frac{t_\mathrm{brake}-t}{\tau_\mathrm{d}}}) \nonumber \\
&\quad -\gamma(t-t_\mathrm{brake}). \label{eq:v0t equation}
\end{align}
For $t=t^*$, we have $v_0(t) = v_0(t^*) = \frac{\gamma}{\eta}$. Therefore, (\ref{eq:v0t equation}) implies
\begin{align}
\frac{\gamma}{\eta}  &= v_0 (t_\mathrm{brake}) + \tau_\mathrm{d}(a_0(t_\mathrm{brake})+\gamma)(1-e^{\frac{t_\mathrm{brake}-t^*}{\tau_\mathrm{d}}}) \nonumber \\
&\quad -\gamma(t^*-t_\mathrm{brake}), \label{eq:v0tstar equation}
\end{align}
which in turn implies
\begin{align}
\frac{t^*}{\tau_\mathrm{d}} + \beta_1 = -\frac{a_0(t_\mathrm{brake})+\gamma}{\gamma} e^{\frac{t_\mathrm{brake}-t^*}{\tau_\mathrm{d}}}. \label{eq:tstar beta1}
\end{align}
By multiplying both sides by $e^{\frac{t^*}{\tau_\mathrm{d}}+\beta_1}$, we obtain
\begin{align}
\left(\frac{t^*}{\tau_\mathrm{d}} + \beta_1\right)e^{\frac{t^*}{\tau_\mathrm{d}}+\beta_1} = \left(\frac{-(a_0(t_\mathrm{brake})+\gamma)e^{\frac{t_\mathrm{brake}}{\tau_\mathrm{d}}+\beta_1}}{\gamma}\right). \label{eq:last step to lamber w function}
\end{align}
Finally, (\ref{eq:tstar calculation}) follows from (\ref{eq:last step to lamber w function}) and the definition of Lambert $W$ function.
 \end{proof}

Notice that by the definition of $t^*$, we have $u_0(t)=-\gamma$ for $t\in [t_{\mathrm{brake}}, t^*)$. Furthermore, $u_0(t)=-\eta v_0(t)$ for $t\geq t^*$. Therefore, it follows from \eqref{eq:p0-def}--\eqref{eq:a0-def} that for $t\geq t^*$, the dynamics of the virtual reference vehicle's velocity $v_0(t)$ follow the second-order linear homogeneous equation given by 
\begin{align}
    \tau_\mathrm{d}\ddot{v}_0(t)+\dot{v}_0(t)+\eta v_0(t) = 0. \label{eq:governing eq}
\end{align}
This equation allows us to obtain $v_0(t)$ as an explicit function of time for $t\geq t^*$, and thus we can obtain $u_0(t)=-\eta v_0(t)$ as an explicit function of time.

We remark that the characteristic equation associated with \eqref{eq:governing eq} is given as $\tau_\mathrm{d} \lambda^2 + \lambda + \eta = 0$. The solutions to \eqref{eq:governing eq} depend on the roots of this characteristic equation. Following \cite{ode1985}, we can consider three cases for the solutions: 1) $\eta < 1/(4\tau_{\mathrm{d}})$ (overdamped case), 2) $\eta = 1/(4\tau_{\mathrm{d}})$ (critically damped case), and 3) $\eta > 1 / (4\tau_{\mathrm d})$ (underdamped case). For the underdamped case, the discriminant of the characteristic equation becomes negative. This results in a pair of complex conjugate roots, leading to an oscillatory solution for $v_0(t)$. This is mathematically valid, but the oscillatory behavior is generally considered undesirable for a vehicle's braking system. It would imply that the vehicle's deceleration rate fluctuates, potentially causing dangerous situations during braking, which is suboptimal for stable and safe vehicle platoon control. Therefore, in the following result, we focus only on overdamped and critically damped cases, that is, $\eta \leq 1/(4\tau_{\mathrm{d}})$.

\begin{prop}\label{prop:dynamic behavior} Consider the time constant $t^*$ defined in (\ref{eq:tstar definiton}) in relation to the braking model in \eqref{eq:brake-model}. If $v_0(t_{\mathrm{brake}})>\gamma/\eta$ and $\eta \leq 1/(4\tau_\mathrm{d})$, then the behavior of the desired acceleration $u_0(t)$ for $t \geq t^*$ is determined by the relationship between the deceleration parameter $\eta$ and the time constant $\tau_{\mathrm{d}}$ as follows:
    \begin{enumerate}
        \item If $\eta < \frac{1}{4\tau_\mathrm{d}}$, then for $t\geq t^*$, we have
        $$u_0(t) = - \eta \left(\beta_2 e^{\lambda_1(t-t^*)} + \beta_3 e^{\lambda_2(t-t^*)}\right),$$
        where $\lambda_1= \frac{-1+\sqrt{1-4\eta\tau_\mathrm{d}}}{2\tau_\mathrm{d}}, \lambda_2= \frac{-1-\sqrt{1-4\eta\tau_\mathrm{d}}}{2\tau_\mathrm{d}}$, and $\beta_2 = \frac{a_0(t^*)-\lambda_2\frac{\gamma}{\eta}}{\lambda_1-\lambda_2}$, $\beta_3=\frac{\lambda_1\frac{\gamma}{\eta}-a_0(t^*)}{\lambda_1-\lambda_2}$. 

        \item If $\eta = \frac{1}{4\tau_\mathrm{d}}$, then for $t\geq t^*$, we have
        $$u_0(t) = - \eta \cdot e^{\lambda_3(t-t^*)} \left(\beta_4 + \beta_5(t-t^*)\right),$$
        where $\lambda_3 = -\frac{1}{2\tau_\mathrm{d}}$, $\beta_4 = \frac{\gamma}{\eta}$ and $\beta_5 =a_0(t^*)+2\gamma$.
    \end{enumerate}
\end{prop}
\begin{proof} If $\eta < 1/(4\tau_\mathrm{d})$, then $v_0(t) = \beta_2 e^{\lambda_1(t-t^*)} + \beta_3 e^{\lambda_2(t-t^*)}$ satisfies \eqref{eq:governing eq}. If, on the other hand, $\eta = 1/(4\tau_\mathrm{d})$, then $v_0(t) =  e^{\lambda_3(t-t^*)} \left(\beta_4 + \beta_5(t-t^*)\right)$ satisfies  \eqref{eq:governing eq}. Therefore, the result follows by noting that $u_0(t)=-\eta v_0(t)$ for $t\geq t^*$.
\end{proof}

Proposition~\ref{prop:dynamic behavior} provides a formulation of the braking behavior as an explicit function of time and it allows us to introduce time-discretization in our simulation methods. 

\begin{figure}[t]
\centering 
\includegraphics[width=0.9\columnwidth]{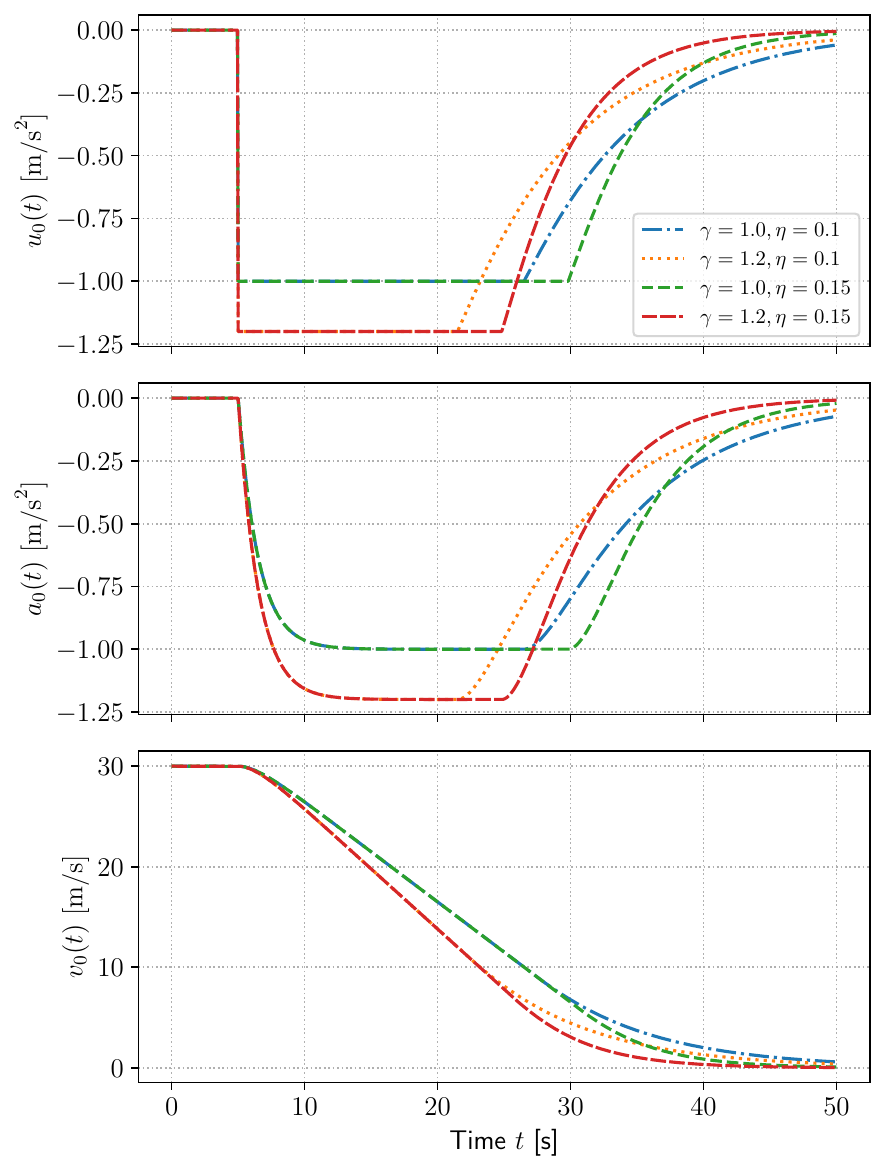}
\caption{Plots of desired acceleration $u_0(t)$, acceleration $a_0(t)$, and velocity $v_0(t)$ with respect to time $t$ for different braking parameters $\gamma$ and $\eta$.}
\label{plot_brake}
\end{figure}

In Figure~\ref{plot_brake}, we show plots of the acceleration and velocity profile of the virtual reference vehicle with a characteristic time constant $\tau_{\mathrm{d}}=1.5$ under a sudden brake at time $t=5$. We consider different values of the brake parameters $\eta$ and $\gamma$. As shown, when the value of $\eta$ is kept constant, a larger $\gamma$ leads to greater deceleration. We also observe that when the value of $\gamma$ is the same, as a result of increasing the value of $\eta$, the value of $t^*$ becomes larger, and the curves of $u_0(t)$ and $a_0(t)$ become steeper after $t^*$. These acceleration and velocity profiles are similar to those observed in real-world  experiments presented in \cite{lyubenov2014} that utilized high-frequency VBOX 3i 100Hz GPS Data Loggers to record emergency braking maneuvers in vehicles (see, e.g., Figure 5 and Figure 6 in \cite{lyubenov2014}).

\subsection{Packet losses}
\label{sec:CommunicationLosses}

In this paper, we model the packet losses in two distinct ways: $\ell$-consecutive packet losses and random packet losses. In both cases, a packet loss is a transient outage of an individual packet transmission at a periodic communication instant. It does not represent a permanent failure of the communication system. When a transmission fails, the vehicle at the receiving side keeps the most recently received desired-acceleration value. 

\subsubsection{$\ell$-consecutive packet losses}

As a first model of packet losses, we consider consecutive
packet losses similar to those typically considered in networked control problems \cite{xiong2007stabilization}. In this model, communication attempts by each vehicle
$i\in\{1,\ldots,n-1\}$ are successful only at times $0,\ell+1,2(\ell+1),\ldots$,
where $\ell\in{\mathbb{N}}$ denotes the number of consecutive packet losses
in between those times. With this setting, the packet loss 
indicator $l_{j}^{i}$ for each vehicle $i\in\{1,\ldots,n-1\}$ is
given by 
\begin{align}
l_{j}^{i} & =\begin{cases}
0, & j\in{\mathcal{J}}\triangleq\{0,\ell+1,2(\ell+1),\ldots\},\\
1, & j\in{\mathbb{N}}_{0}\setminus{\mathcal{J}}.
\end{cases}\label{eq:ell-drop}
\end{align}
 If $\ell$ is large, then the vehicles cannot communicate frequently.
In such a case, if a vehicle is braking and its desired acceleration
value cannot be passed in time to the next vehicle, then a collision
 may happen. 

\subsubsection{Random packet losses}

In the networked control literature, a typical model for wireless
networks is the model characterized by random packet losses \cite{hespanha2007,cetinkaya2016tac}.
In this model, $\{l_{j}^{i}\in\{0,1\}\}_{j\in{\mathbb{N}}_{0}}$ is
a Bernoulli-process with parameter $p\in[0,1]$ indicating the packet
loss probability. Specifically, 
\begin{align*}
 & {\mathbb{P}}[l_{j}^{i}=1]=p,\quad{\mathbb{P}}[l_{j}^{i}=0]=1-p,
\end{align*}
 for $j\in{\mathbb{N}}_{0}$ and $i\in\{1,2,\ldots,n-1\}$. 

We note that a similar packet loss model was used by \cite{acciani2022}
to characterize the lossy communication in vehicle platoons. There
a discrete-time vehicle platoon model was considered. 

\section{Time-Discretized Simulation Approach }
\label{sec:Time-Discretized-Simulation-Appr}
In this paper, we aim to calculate the minimum of minimum inter-vehicle
distances ($d_{\min}^{*}$) in (\ref{eq:d-min-def}) by numerically
finding solutions to the linear system (\ref{eq:linear-system})
for $t\in[0,t_{{\mathrm{end}}}]$. Our method is based on evaluating
the solution at discrete time instants in the set ${\mathcal{T}}\triangleq\{t_{0},t_{1},t_{2},\ldots\}$, where $t_i$, $i\in \mathbb{N}_0$, are called simulation time instants. 

\subsection{Time-discretized braking model}
To facilitate the simulation, we time-discretize the braking model   
 in (\ref{eq:brake-model}) as 
 \begin{align} 
 u_{0}(t) & \triangleq u_{0,k},\quad t\in[s_{k},s_{k+1}),\quad k\in{\mathbb{N}}_{0},\label{eq:discrete-brake-model-1-revised} 
 \end{align} 
 where 
 \begin{align} 
 u_{0,k} & =\begin{cases} 
 0, & kT<t_{{\mathrm{brake}}},\\
 -\gamma, & t_{{\mathrm{brake}}} \leq kT < t^*, \\
- \eta \big(\beta_2 e^{\lambda_1(kT-t^*)} \\
 \,\,\,\,+ \beta_3 e^{\lambda_2(kT-t^*)}\big), &kT \geq t^* \quad \text{and} \quad \eta<\frac{1}{4\tau_\mathrm{d}}, \\
 - \eta e^{\lambda_3(kT-t^*)} \big(\beta_4 \\
 \,\,\,\, + \beta_5(kT-t^*)\big), &kT\geq t^* \quad \text{and} \quad \eta=\frac{1}{4\tau_\mathrm{d}}, 
 \end{cases}\label{eq:discrete-brake-model-2-revised} 
 \end{align} 
 and $\beta_2, \beta_3, \beta_4, \beta_5$ are given in Proposition~\ref{prop:dynamic behavior}, $k\in{\mathbb{N}}_{0}$. Notice that $u_0(t)$ in \eqref{eq:discrete-brake-model-1-revised} is a piecewise constant function of time and it is constant between communication attempt times. This setting allows us to have the same braking behavior in different simulation approaches. This braking model, when combined with the communication model, allows us to obtain the 
 following result.

\begin{lemma} \label{Lemma-key} Consider the linear vehicle platoon
model (\ref{eq:linear-system}) with communication attempt time instants
${\mathcal{S}}=\{s_{0},s_{1},s_{2},\ldots\}$ and the time-discretized
braking model (\ref{eq:discrete-brake-model-1-revised}), (\ref{eq:discrete-brake-model-2-revised}). If ${\mathcal{S}}\subseteq{\mathcal{T}}$ with ${\mathcal{T}}\triangleq\{t_{0},t_{1},t_{2},\ldots\}$,
then the input vector $u(t)$ in (\ref{eq:linear-system}) satisfies
$u(t)=u(t_{k})$ for $t\in[t_{k},t_{k+1})$.\end{lemma} \begin{proof}Since
${\mathcal{S}}\subseteq{\mathcal{T}}$, there exists $j\in{\mathbb{N}}_{0}$
such that $[t_{k},t_{k+1})\subseteq[s_{j},s_{j+1})$. Let $\widetilde{u}_{i,j}\triangleq\hat{u}_{i}(s_{j})$
for $i\in\{1,2,\ldots,n-1\}$. Notice that by (\ref{eq:hat-u-mid})
and (\ref{eq:ui-piecewise-constant}), we have $\hat{u}_{i}(t)=\widetilde{u}_{i,j}$
for $t\in[s_{j},s_{j+1})$. Since $[t_{k},t_{k+1})\subseteq[s_{j},s_{j+1})$,
it follows that $\hat{u}_{i}(t)=\widetilde{u}_{i,j}$ for $t\in[t_{k},t_{k+1})$.
Finally, by definition of $u(t)$ in (\ref{eq:u-def}) and the time-discretized
braking model (\ref{eq:discrete-brake-model-1-revised}), (\ref{eq:discrete-brake-model-2-revised}),
we obtain
\begin{align*}
u(t) & =\left[u_{0,k},\widetilde{u}_{1,j},\widetilde{u}_{2,j},\ldots,\widetilde{u}_{(n-1),j}\right]^{\top}
\end{align*}
for $t\in[t_{k},t_{k+1})$, which completes the proof. \end{proof}

Observe that if the set ${\mathcal{S}}=\{s_{0},s_{1},s_{2},\ldots\}$
of communication attempt time instants satisfies ${\mathcal{S}}\subseteq{\mathcal{T}}\triangleq\{t_{0},t_{1},t_{2},\ldots\}$,
then by Lemma~\ref{Lemma-key}, the solutions of (\ref{eq:linear-system})
under the braking model (\ref{eq:discrete-brake-model-1-revised}), (\ref{eq:discrete-brake-model-2-revised}),
can be computed by 
\begin{align}
&x(t_{0})  =x(0),\label{eq:x0}\\
 &x(t) =e^{A_{{\mathrm{c}}}(t-t_{k})}x(t_{k}) +\int_{0}^{t-t_{k}}e^{A_{{\mathrm{c}}}(t-t_{k}-\tau)}{\mathrm{d}}\tau B_{{\mathrm{c}}}u(t_{k}),\label{eq:xt}
\end{align}
 for $t\in[t_{k},t_{k+1}]$, $k\in{\mathbb{N}}_{0}$. This calculation
is possible, because $u(t)$ is constant in the interval $t\in[t_{k},t_{k+1})$
as implied by Lemma~\ref{Lemma-key}.

A consequence of (\ref{eq:x0}), (\ref{eq:xt}) together with the
definition of $x(t)$ in (\ref{eq:x-def}) and $d_{i}(t)$ in (\ref{eq:di-def})
is that, for each $i\in\{2,3,\ldots,n\}$, 
\begin{align}
d_{i}(t_{0}) & =q_{i}^{\top}x(0)-L_{i},\label{eq:d0}\\
d_{i}(t) & =q_{i}^{\top}e^{A_{{\mathrm{c}}}(t-t_{k})}x(t_{k})\nonumber \\
 & \,\,+q_{i}^{\top}\int_{0}^{t-t_{k}}e^{A_{{\mathrm{c}}}(t-t_{k}-\tau)}{\mathrm{d}}\tau B_{{\mathrm{c}}}u(t_{k})-L_{i},\label{eq:dt}
\end{align}
 for $t\in[t_{k},t_{k+1}]$, $k\in{\mathbb{N}}_{0}$. 

\subsection{Proposed simulation approach} \label{PSA}

In our simulation approach we evaluate the solutions $x(t_{k})$,
$k\in{\mathbb{N}}_{0}$, by using (\ref{eq:xt}) and then we calculate
$d_{i}(t_{k})$, $k\in{\mathbb{N}}_{0}$, by using (\ref{eq:di-equality}).
Since $d_{i}$ values are calculated only at times $t_{0},t_{1},t_{2},\ldots$,
there is a possibility that vehicles may get dangerously close to
each other in between these simulation times even if $d_i(t_{0}),d_i(t_{1}),d_i(t_{2}),\ldots$
are sufficiently large. The following result indicates that if the
lengths of intervals between times $t_{0},t_{1},t_{2},\ldots$ are
sufficiently small, then for $t\in[t_{k},t_{k+1})$, the inter-vehicle
distance $d_{i}(t)$ does not deviate much from $d_{i}(t_{k})$. 

\begin{figure}[t]
\centering 
\includegraphics[width=0.9\columnwidth]{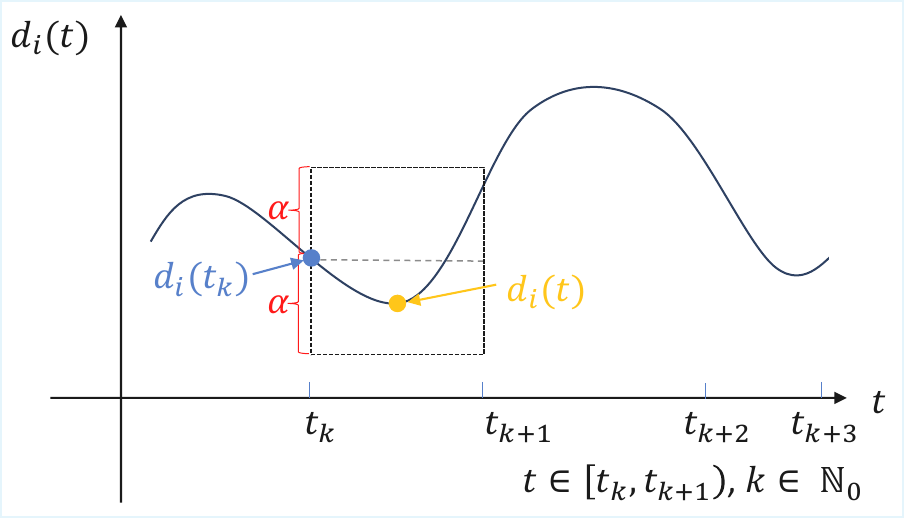}
\caption{An illustration of inter-vehicle distance $d_i(t)$ with respect to time $t \in [t_k,t_{k+1})$. }
\label{plot_di_illustration}
\end{figure}

\begin{theorem} \label{Theorem-Main} Consider the linear vehicle
platoon model (\ref{eq:linear-system}) with communication attempt
time instants ${\mathcal{S}}=\{s_{0},s_{1},s_{2},\ldots\}$ and the
time-discretized braking model (\ref{eq:discrete-brake-model-1-revised}), \eqref{eq:discrete-brake-model-2-revised}.
If ${\mathcal{S}}\subseteq{\mathcal{T}}$ with   ${\mathcal{T}}\triangleq\{t_{0},t_{1},t_{2},\ldots\}$ and there exists $\alpha>0$
such that 
\begin{align}
 & t_{k+1}-t_{k}\leq\ln\left(\frac{\alpha}{\sqrt{2}(\|x(t_{k})\|+\frac{\|B_{{\mathrm{c}}}u(t_{k})\|}{\|A_{{\mathrm{c}}}\|})}+1\right)/\|A_{{\mathrm{c}}}\|,\label{eq:tk-condition1}
\end{align}
for every $k\in{\mathbb{N}}_{0}$, then 
\begin{align}
|d_{i}(t)-d_{i}(t_{k})| & \leq\alpha,\quad t\in[t_{k},t_{k+1}),\label{eq:di-bound}
\end{align}
for every $k\in{\mathbb{N}}_{0}$ and every $i\in\{2,3,\ldots,n\}$.
\end{theorem}

Theorem~\ref{Theorem-Main} is an improved version of Theorem~3 of the conference version \cite{chen2025sice} of this paper. It provides a less restrictive condition compared to that presented in \cite{chen2025sice} and its proof is given in the Appendix. It indicates that if the lengths of the
intervals between the simulation time instants $t_{0},t_{1},t_{2}\ldots,$
are sufficiently small so that (\ref{eq:tk-condition1}) holds with
a scalar $\alpha>0$, then the inter-vehicle distances do not deviate
more than $\alpha$ units in each time interval $[t_{k},t_{k+1})$. This point is illustrated in Figure~\ref{plot_di_illustration}. Furthermore, this shows that if the simulation results indicate 
\begin{align*}
d_{i}(t_{k}) & \geq\theta,\quad k\in{\mathbb{N}}_{0},
\end{align*}
for some $\theta>0$, then $d_{i}(t)\geq\theta-\alpha$ holds by (\ref{eq:di-bound}). If, for example, $d_i(t_k)>\alpha$, then $d_i(t)>0$ holds for all $t\in [t_k,t_{k+1})$, which indicates that there is no crash between $(i-1)$th and $i$th vehicles within that interval, even though inter-vehicle distances are not explicitly calculated in the strict interval between simulation times.  

Similarly,
if we compute 
\begin{align}
d_{\min} & \triangleq\min_{i\in\{2,\ldots,n\}}\min_{k\in\{0,\ldots,k_{{\mathrm{end}}}\}}d_{i}(t_{k})\label{eq:dmin-discrete}
\end{align}
 with $k_{\mathrm{end}}\in\mathbb{N}_0$ denoting the total number of simulation time instants and observe that 
\begin{align*}
d_{\min} & \geq\theta,
\end{align*}
 then $d_{\min}^{\star}$ given in (\ref{eq:d-min-def}) with $t_{{\mathrm{end}}}=t_{k_{{\mathrm{end}}}}$ satisfies
\begin{align}
d_{\min}^{\star} & \geq\theta-\alpha.\label{eq:d-min-satisfies-beta-minus-alpha-bound}
\end{align}
 This is because by (\ref{eq:di-bound}), we have $|d_{\min}^{\star}-d_{\min}|\leq\alpha$.
We remark that if $d_{\min}>\alpha$, it is guaranteed that there
is no collision in the platoon, since in this case it is guaranteed
that $d_{\min}^{*}>0$. 

\subsection{A faster simulation approach with lifted states}
In the following, we show that a faster simulation is possible under additional conditions. In particular, we provide a faster approach compared to that in Theorem~\ref{Theorem-Main}, by considering alternative dynamics with the lifted state vector  
\begin{align}
\tilde{x}(t)\triangleq[x^{\top}(t),u^{\top}(t)]^{\top}, \quad t\geq0. \label{eq:xtildet}
\end{align}
This vector combines the original state $x(t)$ and input $u(t)$ of the system \eqref{eq:linear-system} into a single vector. Since the input $u(t)$ remains constant in the interval $[t_{k},t_{k+1})$, it follows from \eqref{eq:linear-system} that
\begin{align}
\dot{\tilde{x}}(t) &= \tilde{A}\tilde{x}(t), \quad t\in[t_{k},t_{k+1}), \label{eq:lifted-system}
\end{align}
where $\tilde{A}\in\mathbb{R}^{(4+7n)\times(4+7n)}$ is given in block-matrix form as 
\begin{align}
\tilde{A} &\triangleq \begin{bmatrix} A_\mathrm{c} & B_\mathrm{c} \\ 0 & 0 \end{bmatrix}. \label{eq:mu-Atilde-def}
\end{align}
Now, we define two matrices $b_1\in \mathbb{R}^{(4+7n)\times(3+6n)}$  and $b_2\in\mathbb{R}^{(4+7n)\times(1+n)}$  in block-matrix form as 
\begin{align*}
b_1 &\triangleq \begin{bmatrix} I_{3+6n} \\ 0_{(1+n)\times(3+6n)} \end{bmatrix}, &b_2 \triangleq \begin{bmatrix} 0_{(3+6n)\times(1+n)} \\ I_{1+n} \end{bmatrix}.
\end{align*} 
Notice that 
\begin{align}
    x(t)=b_1^\top\tilde{x}(t),\quad u(t)=b_2^\top\tilde{x}(t). \label{eq:x-xtilde-relation}
\end{align}
We are now ready to provide an alternative to Theorem~\ref{Theorem-Main}. 
\begin{theorem} \label{Theorem-7}
Consider the linear vehicle
platoon model (\ref{eq:linear-system}) with communication attempt
time instants ${\mathcal{S}}=\{s_{0},s_{1},s_{2},\ldots\}$ and the
time-discretized braking model (\ref{eq:discrete-brake-model-1-revised}), \eqref{eq:discrete-brake-model-2-revised}.
Suppose $\mu(\tilde{A})>0$ holds for $\tilde{A}$ defined in \eqref{eq:mu-Atilde-def}. If ${\mathcal{S}}\subseteq{\mathcal{T}}$ with   ${\mathcal{T}}\triangleq\{t_{0},t_{1},t_{2},\ldots\}$, and there exists $\alpha>0$
such that
\begin{align}
 & t_{k+1}-t_{k} \leq\ln\left(\frac{\mu(\tilde{A})\alpha}{\varphi\|\tilde{x}(t_{k})\|}+1\right)/\mu(\tilde{A}), \label{eq:tk-condition7}
\end{align}
for every $k\in{\mathbb{N}}_{0}$, where
\begin{align}
\varphi \triangleq \max_{i\in\{2,\ldots n\}}\|q_{i}^{\top}(A_\mathrm{c}b_{1}^{\top}+B_\mathrm{c}b_{2}^{\top})\|, \label{eq:varphi_2}
\end{align}
then (\ref{eq:di-bound}) holds for every $k\in{\mathbb{N}}_{0}$ and every $i\in\{2,\ldots,n\}$.
\end{theorem}


\begin{proof}
The differences in inter-vehicle distances for  $t\in [t_k, t_{k+1})$ and their upper bounds are calculated by applying the Fundamental Theorem of Calculus as 
\begin{align}
|d_i(t)-d_i(t_k)| = \left|\int_{t_k}^{t}\frac{\mathrm{d}}{\mathrm{d}s}d_i(s)\mathrm{d}s\right|\leq\int_{t_k}^{t}\left|\frac{\mathrm{d}}{\mathrm{d}s}d_i(s)\right|\mathrm{d}s. \label{eq:intervehicledistance_int}
\end{align}
 By \eqref{eq:linear-system}, \eqref{eq:di-equality}, and \eqref{eq:x-xtilde-relation}, we obtain
\begin{align}
\frac{\mathrm{d}}{\mathrm{d}s}d_i(s) &= q_i^\top \dot{x}(s) = q_i^\top (A_\mathrm{c}x(s)+B_\mathrm{c}u(s))  \nonumber \\ 
 &=q_{i}^{\top}(A_\mathrm{c}b_{1}^{\top}\tilde{x}(s)+B_\mathrm{c}b_{2}^{\top}\tilde{x}(s)) \nonumber \\ 
&= q_{i}^{\top}(A_\mathrm{c}b_{1}^{\top}+B_\mathrm{c}b_{2}^{\top})\tilde{x}(s). \label{eq:derivative_di_7}
\end{align}
Then applying Cauchy-Schwarz inequality, we get 
\begin{align}
\left|\frac{\mathrm{d}}{\mathrm{d}s}d_i(s)\right|\leq \|q_{i}^{\top}(A_\mathrm{c}b_{1}^{\top}+B_\mathrm{c}b_{2}^{\top})\|\|\tilde{x}(s)\|. \label{eq:cauchy-schwarz inequality theorem7_old}
\end{align}
Here, we want to bound $\|\tilde{x}(s)\|$. Note that \eqref{eq:lifted-system} implies that $\tilde{x}(s)=e^{\tilde{A}(s-t_k)}\tilde{x}(t_k)$ for $s\in[t_k,t_{k+1})$. Therefore, by using the properties of the logarithmic matrix norm induced by the Euclidean norm (see Theorem 24 in \cite{desoer2009feedback} and Proposition 2.1 in \cite{soderlind2006logarithmic}), we have
\begin{align}
\|\tilde{x}(s)\| = \|e^{\tilde{A}(s-t_k)}\tilde{x}(t_k)\| &\leq \|e^{\tilde{A}(s-t_k)}\|\|\tilde{x}(t_k)\|   \nonumber \\
&\leq  e^{\mu(\tilde{A})(s-t_k)}\|\tilde{x}(t_k)\|. \label{eq:tildexs}
\end{align}
By substituting (\ref{eq:tildexs}) to (\ref{eq:cauchy-schwarz inequality theorem7_old}) and using the constant $\varphi$, we get 
\begin{align}
\left|\frac{\mathrm{d}}{\mathrm{d}s}d_i(s)\right| &\leq \|q_{i}^{\top}(A_\mathrm{c}b_{1}^{\top}+B_\mathrm{c}b_{2}^{\top})\| e^{\mu(\tilde{A})(s-t_k)}\|\tilde{x}(t_k)\| \nonumber \\
& \leq \varphi e^{\mu(\tilde{A})(s-t_k)}\|\tilde{x}(t_k)\|.
\label{eq:cauchy-schwarz inequality theorem2}
\end{align}
Now, by \eqref{eq:intervehicledistance_int} and (\ref{eq:cauchy-schwarz inequality theorem2}), we arrive at
\begin{align}
|d_i(t)-d_i(t_k)| &\leq \int_{t_k}^{t}\varphi e^{\mu(\tilde{A})(s-t_k)}\|\tilde{x}(t_{k})\|\mathrm{d}s \nonumber \\
&=\varphi\|\tilde{x}(t_{k})\|\int_{t_k}^{t}e^{\mu(\tilde{A})(s-t_k)}\mathrm{d}s \nonumber \\
&= \frac{\varphi\|\tilde{x}(t_{k})\|} 
{\mu(\tilde{A})}(e^{\mu(\tilde{A})(t-t_k)}-1).
\label{eq:dit_minus_ditk_7}
\end{align}
Furthermore, for $t\in[t_{k},t_{k+1})$, we have $t-t_{k}\leq t_{k+1}-t_{k}$. Therefore, by (\ref{eq:tk-condition7}), we obtain
\begin{align}
e^{\mu(\tilde{A})(t-t_k)} &\leq e^{\mu(\tilde{A})(t_{k+1}-t_k)} \leq e^{\mu(\tilde{A})\ln\left(\frac{\mu(\tilde{A})\alpha}{\varphi\|\tilde{x}(t_{k})\|}+1\right)/\mu(\tilde{A})} \nonumber \\
&= \frac{\mu(\tilde{A})\alpha}{\varphi\|\tilde{x}(t_{k})\|}+1. 
\label{eq:alpha_constraint_7}
\end{align}
It then follows from (\ref{eq:dit_minus_ditk_7}) and (\ref{eq:alpha_constraint_7}) that
\begin{align*}
|d_i(t) - d_i(t_k)| \leq \frac{\varphi\|\tilde{x}(t_{k})\|}{\mu(\tilde{A})} \frac{\mu(\tilde{A})\alpha}{\varphi\|\tilde{x}(t_{k})\|} = \alpha,
\end{align*}
which completes the proof.
\end{proof}

Compared with Theorem~\ref{Theorem-Main}, the primary advantage of Theorem~\ref{Theorem-7} is that it enhances computational efficiency. Specifically, in practice, \eqref{eq:tk-condition7} allows for a larger simulation time gap $t_{k+1} - t_k$ compared to \eqref{eq:tk-condition1}, meaning that the simulation can be completed in fewer discrete steps. Notice that, while the simulation running time is reduced, we still have the same guaranteed error bound $\alpha$  in (\ref{eq:di-bound}). This point is illustrated in Section~\ref{sec:Numerical-Example}.
\begin{remark} \label{RemarkMu}
In comparison with Theorem~\ref{Theorem-Main}, Theorem~\ref{Theorem-7} has an additional condition.  Theorem~\ref{Theorem-7} requires $\mu(\tilde{A})>0$. Notice that this is satisfied when the matrix $N$ in $A_\mathrm{c}$ has at least one positive eigenvalue. This is because $N$ having a positive eigenvalue implies that $A_\mathrm{c}$ also has a positive eigenvalue, since $A_\mathrm{c}$ is block lower-triangular and $N$ is one of the blocks on the diagonal. Moreover, $A_{\mathrm{c}}$ having a positive eigenvalue implies $A_\mathrm{c}+A^\top_\mathrm{c}$ also has a positive eigenvalue, that is $\lambda_{\mathrm{max}}(A_\mathrm{c}+A_\mathrm{c}^\top) > 0$. Therefore, there exists $x\in \mathbb R^{3+6n}\setminus \{0\}$ such that 
\begin{align*}
    0 & <x^\top(A_\mathrm{c}+A_\mathrm{c}^\top)x =[x^\top,0_{1\times(1+n)}](\tilde A + \tilde A ^\top )[x^\top,0_{1\times(1+n)}]^\top,
\end{align*}
which in turn implies that $\tilde{A}+\tilde{A}^\top$ has at least one positive eigenvalue. Hence, $\mu(\tilde A)=\lambda_{\max} ((\tilde{A}+ \tilde{A} ^\top)/2)>0$. In numerical scenarios that we consider in Section~\ref{sec:Numerical-Example}, $\mu(\tilde{A})>0$ is always satisfied.
\end{remark}

\subsection{Design of simulation times} \label{subsection:design of simulation times}
Excluding the condition $\mu(\tilde{A})>0$ discussed in Remark~\ref{RemarkMu}, Theorems~\ref{Theorem-Main} and \ref{Theorem-7} both impose two conditions to guarantee that the bound in (\ref{eq:di-bound})
holds. The first of these conditions is ${\mathcal{S}}\subseteq{\mathcal{T}}$. The second one is 
 (\ref{eq:tk-condition1}) for Theorem~\ref{Theorem-Main} and \eqref{eq:tk-condition7} for Theorem~\ref{Theorem-7}. These two conditions can be simultaneously satisfied by designing
the simulation times $t_{0},t_{1},t_{2}\ldots,$ in a special way.
Specifically, we consider a large integer $\overline{N}\in{\mathbb{N}}$.
Now, for each $k\in{\mathbb{N}}_{0}$, we define
\begin{align}
\nu_{k}^{(1)}\triangleq & \left\lfloor \frac{\overline{N}}{T}\frac{\ln\left(\frac{\alpha}{\sqrt{2}(\|x(t_{k})\|+\|B_{{\mathrm{c}}}u(t_{k})\|/\|A_{{\mathrm{c}}}\|)}+1\right)}{\|A_{{\mathrm{c}}}\|}\right\rfloor, \label{eq:nuk3-def} \\
\nu_{k}^{(2)}\triangleq & \left\lfloor \frac{\overline{N}}{T} \frac{\ln\left(\frac{\mu(\tilde{A})\alpha}{\varphi\|\tilde{x}(t_{k})\|}+1\right)}{\mu(\tilde{A})}\right\rfloor. \label{eq:nuk8-def}
\end{align}
Then for Theorem~$i$ with $i$ being either $1$ or $2$, we set 
\begin{align}
t_{0} & =0,\label{eq:tset-0}\\
t_{k+1} & = t_k + \frac{T}{\overline{N}} \min\left\{ \left\lfloor \frac{\overline{N}}{T}
\left( \left( \left\lfloor \frac{t_k}{T} \right\rfloor + 1 \right) T - t_k \right) \right\rfloor,\, \nu_{k}^{(i)} \right\}, \label{eq:tset-k}
\end{align}
for $k\in\mathbb{N}_0$. It follows that (\ref{eq:tk-condition1}) holds as long as $\overline{N}$
is sufficiently large so that $1\leq\nu_{k}^{(1)}\leq\overline{N}$ holds
for every $k\in{\mathbb{N}}_{0}$. Similarly, (\ref{eq:tk-condition7}) holds as long as $\overline{N}$
is sufficiently large so that $1\leq\nu_{k}^{(2)}\leq\overline{N}$ holds
for every $k\in{\mathbb{N}}_{0}$. Furthermore, since $s_{0}=0$ and
$s_{j+1}-s_{j}=T$ for $j\in{\mathbb{N}}_{0}$, it follows from (\ref{eq:tset-0})
and (\ref{eq:tset-k}) that ${\mathcal{S}}\subseteq{\mathcal{T}}$.

In Section~\ref{sec:Numerical-Example}, we utilize the proposed simulation methods to study the behavior of vehicles in a platoon under sudden braking and packet losses. However, before that, in the next section, we aim to expand our simulation approach to analyze fuel efficiency of vehicle platoons. 

\section{Fuel Efficiency Analysis via Simulation}
\label{sec:FuelEfficiency}

While safety is the primary concern in vehicle platooning, fuel efficiency is also a key motivation for its deployment. The aerodynamic drag, which is the major resistance at highway speeds, is significantly influenced by the inter-vehicle distance $d_i(t)$. To evaluate the performance of a vehicle platoon in terms of efficiency, we aim to quantify the fuel savings enabled by the platoon formation.

In what follows, we first use the findings of past research \cite{hussein2021,barhoumi2025} to characterize the fuel savings achieved by each vehicle within the interval between two simulation times. Then, we obtain error bounds that our approach can guarantee when these fuel savings are quantified by simulation. We note that past research does not provide such bounds.

\subsection{Fuel savings of each vehicle in a platoon}

First, based on the aerodynamic drag model in equation (8) of \cite{barhoumi2025}, the drag force $F_{\mathrm{D},i}(t)$ acting on the $i$th vehicle (with $i\in \{2,3,\ldots,n\}$) at time $t$ is defined as
\begin{equation}
F_{\mathrm{D},i}(t) \triangleq \frac{1}{2} \rho K_\mathrm{vfa} C_{\mathrm{D},i}(d_i(t)) v_i^2(t), \label{eq:drag force}
\end{equation}
where $\rho>0$ denotes the air density, $K_\mathrm{vfa}>0$ is the vehicle frontal area (assumed for simplicity to be the same for all vehicles), $v_i(t)$ is the velocity of the $i$th vehicle as characterized in \eqref{eq:vi-def-first}, and $C_{\mathrm{D},i}(d_i(t))>0$ is the drag coefficient.
As described in equation (9) of \cite{barhoumi2025}, the drag coefficient $C_{\mathrm{D},i}(d_i(t))$ is not constant but varies as a function of the inter-vehicle distance $d_i(t)$.

To explicitly characterize the aerodynamic drag coefficient $C_{\mathrm{D},i}$ as a function of the inter-vehicle distance $d_i(t)$, we use the rational polynomial model proposed by \cite{hussein2021}. This model is obtained from the wind tunnel data and captures the complex nonlinear aerodynamic interactions. In this model, the aerodynamic drag coefficient for the $i$-th following vehicle is expressed as
\begin{equation}
C_{\mathrm{D},i}(d_i(t)) = C_{\mathrm{B}}\mathcal{F}_i(d_i(t)), \label{eq:drag_coeff_general}
\end{equation}
where $C_{\mathrm{B}}>0$ denotes the baseline drag coefficient of a solitary vehicle driving by itself and $\mathcal{F}_i(d_i(t))$ is the drag reduction factor defined for each vehicle $i$ by a third-order rational polynomial function (see Section III of \cite{hussein2021}) as 
\begin{align}
&\mathcal{F}_i(d) 
\triangleq 
\begin{cases} 
  \frac{\sum_{j=0}^{3}\hat{w}_{i,j}d^j}{\sum_{j=0}^{3}\tilde{w}_{i,j}d^j}, & 0 < d \leq d_{\mathrm{max},i}, \\
1, & d > d_{\mathrm{max},i},
\end{cases}
\label{eq:rational_poly}
\end{align}
where $d_{\mathrm{max},i}>0$ represents the maximum distance in front of $i$th vehicle where drag is considered to have an effect. Furthermore,  $\hat{w}_{i,j}\in \mathbb{R}$ and $\tilde{w}_{i,j}\in \mathbb{R}$, for $j\in \{0,1,2,3\}$, are empirical coefficients obtained, for example, for light-duty vehicles, as presented in Table III of \cite{hussein2021}. Notice that $\mathcal{F}_i(d_i(t))\in [0,1]$ always holds. 

Next, by following the equation (10) of \cite{barhoumi2025}, the power savings $P_{\mathrm{save},i}(t)$ for the $i$th vehicle, is derived as the product of the drag force reduction (difference between driving alone and driving in platoon) and velocity as
\begin{equation}
P_{\mathrm{save},i}(t) = \left( \frac{1}{2} \rho K_\mathrm{vfa} C_{\mathrm{B}}v_i^2(t) - F_{\mathrm{D},i}(t) \right) v_i(t). \label{eq:pw saving}
\end{equation}
By substituting (\ref{eq:drag force}) into (\ref{eq:pw saving}), we get
\begin{equation}
P_{\mathrm{save},i}(t) = \frac{1}{2} \rho K_\mathrm{vfa} \left( C_{\mathrm{B}}- C_{\mathrm{D},i}(d_i(t)) \right) v_i^3(t). \label{eq:Psave-expanded}
\end{equation}

Finally, we use the fuel consumption relationship from equation (11) of \cite{barhoumi2025} to characterize the fuel saving $m_{\mathrm{save},i}$ of the $i$th vehicle over the simulation time steps [$t_k$, $t_{k+1}$) as
\begin{equation}
m_{\mathrm{save},i,k} \triangleq  \int_{t_k}^{t_{k+1}} \frac{P_{\mathrm{save},i}(\tau)}{Q_{\mathrm{LHV}}\eta_\mathrm{engine}} \mathrm{d}\tau, \label{eq:fuel_integral_tk-tk+1}
\end{equation}
where $\eta_\mathrm{engine}>0$ is the constant engine efficiency and $Q_{\mathrm{LHV}}>0$ is the fuel lower heating value.
Similarly, the total amount of fuel savings $M_{\mathrm{save},i}$ over the entire simulation duration $[0, t_{\mathrm{end}}]$ is computed as 
\begin{equation}
M_{\mathrm{save},i} \triangleq \int_{0}^{t_{\mathrm{end}}} \frac{P_{\mathrm{save},i}(\tau)}{Q_{\mathrm{LHV}}\eta_\mathrm{engine}} \mathrm{d}\tau. \label{eq:fuel_integral_0-tend}
\end{equation}

In our simulation framework presented in Section~\ref{sec:Time-Discretized-Simulation-Appr}, inter-vehicle distances $d_i(t)$ and velocities $v_i(t)$ are only computed at simulation times. As a result, the integrals in \eqref{eq:fuel_integral_tk-tk+1} cannot be evaluated exactly, and thus, fuel saving $m_{\mathrm{save},i,k}$ cannot be computed exactly. Instead, we approximate the integral in \eqref{eq:fuel_integral_tk-tk+1} and obtain approximate fuel savings as  
\begin{equation}
\tilde{m}_{\mathrm{save},i,k} \triangleq  (t_{k+1}-t_{k}) \frac{P_{\mathrm{save},i}(t_k)}{Q_{\mathrm{LHV}}\eta_\mathrm{engine}}. \label{eq:fuel_integral_tk-tk+1_approximate}
\end{equation}
The error introduced by this approximation is defined as
\begin{align}
    E_{i,k}\triangleq \tilde m_{\mathrm{save},i,k}-m_{\mathrm{save},i,k}.\label{eq:error_integral_def}
\end{align}
To ensure the reliability of our approach, in the following, we derive a theoretical upper bound for this error.

\subsection{Guaranteed error bound for quantification of fuel savings in a platoon}
Our goal in this section is to obtain a guaranteed bound on the absolute value of the fuel savings approximation error (i.e., an upper bound of $|E_{i,k}|$), by utilizing the bounds that we obtained in Theorem~\ref{Theorem-7}.

As a first step, we want to establish a bound on the difference between the velocity at a given simulation time $t_k$ and the velocity at any time $t\in [t_k,t_{k+1})$.

Similar to (\ref{eq:di-equality}), we have 
\begin{align}
v_{i}(t) & =z_{i}^{\top}x(t),\label{eq:vi-equality}
\end{align}
where $z_{i}\in{\mathbb{R}}^{3+6n}$ is a vector with the $j$th entry given as 
\begin{align}
z_{i,j} & \triangleq\begin{cases}
1, & {\mathrm{if}}\,\,j=6i + 1,\\
0, & {\mathrm{otherwise}},
\end{cases}\quad j\in\{1,2,\ldots,3+6n\}.\label{eq:zij-def}
\end{align}
By following the techniques introduced in Section~\ref{PSA} and  utilizing \eqref{eq:vi-equality} together with Theorem~\ref{Theorem-7}, we obtain a result that provides a bound on the deviation of $v_i(t)$ from $v_i(t_k)$. Before we state this result, we summarize the conditions of Theorem~2 in the following assumption, which will be used throughout the section. 
\begin{assumption} \label{Assumption-Conditions-of-Theorem}
    Considering the linear vehicle
platoon model (\ref{eq:linear-system}) with communication attempt
time instants ${\mathcal{S}}=\{s_{0},s_{1},s_{2},\ldots\}$ and the
time-discretized braking model (\ref{eq:discrete-brake-model-1-revised}), \eqref{eq:discrete-brake-model-2-revised}, it holds that $\mu(\tilde{A})>0$ for $\tilde{A}$ defined in \eqref{eq:mu-Atilde-def},  ${\mathcal{S}}\subseteq{\mathcal{T}}$ with   ${\mathcal{T}}\triangleq\{t_{0},t_{1},t_{2},\ldots\}$, and there exists $\alpha>0$
such that (\ref{eq:tk-condition7}) holds for every $k\in{\mathbb{N}}_{0}$.
\end{assumption}

\begin{theorem} \label{Theorem-vit}
Suppose Assumption~\ref{Assumption-Conditions-of-Theorem} holds. Then for every $k\in{\mathbb{N}}_{0}$ and every $i\in\{2,3,\ldots,n\}$, it holds that
\begin{align}
|v_{i}(t)-v_{i}(t_{k})| & \leq \frac{\psi}{\varphi}\alpha,\quad t\in[t_{k},t_{k+1}),\label{eq:vi-bound}
\end{align}
where
\begin{align}
\psi \triangleq \max_{i\in\{2,\ldots n\}}\|z_{i}^{\top}(A_\mathrm{c}b_{1}^{\top}+B_\mathrm{c}b_{2}^{\top})\|, 
\label{eq:varpsi}
\end{align}
and $\varphi$ is given in (\ref{eq:varphi_2}).
\end{theorem}
\begin{proof}
The differences in velocity for  $t\in [t_k, t_{k+1})$ can be calculated by applying the Fundamental Theorem of Calculus as 
\begin{align}
|v_i(t)-v_i(t_k)| = \left|\int_{t_k}^{t}\frac{\mathrm{d}}{\mathrm{d}s}v_i(s)\mathrm{d}s\right|\leq\int_{t_k}^{t}\left|\frac{\mathrm{d}}{\mathrm{d}s}v_i(s)\right|\mathrm{d}s. \label{eq:velocity_int}
\end{align}
 By using \eqref{eq:linear-system}, \eqref{eq:x-xtilde-relation}, and \eqref{eq:vi-equality}, we get
\begin{align}
\frac{\mathrm{d}}{\mathrm{d}s}v_i(s) &= z_i^\top \dot{x}(s) = z_i^\top (A_\mathrm{c}x(s)+B_\mathrm{c}u(s))  \nonumber \\ 
 &=z_{i}^{\top}(A_\mathrm{c}b_{1}^{\top}\tilde{x}(s)+B_\mathrm{c}b_{2}^{\top}\tilde{x}(s)) \nonumber \\ 
&= z_{i}^{\top}(A_\mathrm{c}b_{1}^{\top}+B_\mathrm{c}b_{2}^{\top})\tilde{x}(s). \label{eq:derivative_di_8}
\end{align}
Then by applying Cauchy-Schwarz inequality, we obtain 
\begin{align}
\left|\frac{\mathrm{d}}{\mathrm{d}s}v_i(s)\right|\leq \|z_{i}^{\top}(A_\mathrm{c}b_{1}^{\top}+B_\mathrm{c}b_{2}^{\top})\|\|\tilde{x}(s)\|. \label{eq:cauchy-schwarz inequality theorem3_old}
\end{align}
In this equation, we want to bound the $\|\tilde{x}(s)\|$ term. Note that \eqref{eq:lifted-system} implies that $\tilde{x}(s)=e^{\tilde{A}(s-t_k)}\tilde{x}(t_k)$ for $s\in[t_k,t_{k+1})$. Therefore, by using the property of the logarithmic matrix norm (see Theorem 24 in \cite{desoer2009feedback}), we have (\ref{eq:tildexs}).
By substituting (\ref{eq:tildexs}) to (\ref{eq:cauchy-schwarz inequality theorem3_old}) and using the constant $\psi$, we get 
\begin{align}
\left|\frac{\mathrm{d}}{\mathrm{d}s}v_i(s)\right| &\leq \|z_{i}^{\top}(A_\mathrm{c}b_{1}^{\top}+B_\mathrm{c}b_{2}^{\top})\| e^{\mu(\tilde{A})(s-t_k)}\|\tilde{x}(t_k)\| \nonumber \\
& \leq \psi e^{\mu(\tilde{A})(s-t_k)}\|\tilde{x}(t_k)\|.
\label{eq:cauchy-schwarz inequality theorem7}
\end{align}
Now, by \eqref{eq:velocity_int} and (\ref{eq:cauchy-schwarz inequality theorem7}), we arrive at
\begin{align}
&|v_i(t)-v_i(t_k)| \leq \int_{t_k}^{t}\psi e^{\mu(\tilde{A})(s-t_k)}\|\tilde{x}(t_{k})\|\mathrm{d}s \nonumber \\
&\quad =\psi\|\tilde{x}(t_{k})\|\int_{t_k}^{t}e^{\mu(\tilde{A})(s-t_k)}\mathrm{d}s = \frac{\psi\|\tilde{x}(t_{k})\|} 
{\mu(\tilde{A})}(e^{\mu(\tilde{A})(t-t_k)}-1).
\label{eq:vit_minus_vitk}
\end{align}
Furthermore, for $t\in[t_{k},t_{k+1})$, we have $t-t_{k}\leq t_{k+1}-t_{k}$. Therefore, by (\ref{eq:tk-condition7}), we obtain (\ref{eq:alpha_constraint_7}).
It then follows from (\ref{eq:vit_minus_vitk}) and (\ref{eq:alpha_constraint_7}) that
\begin{align*}
|v_i(t) - v_i(t_k)| \leq \frac{\psi\|\tilde{x}(t_{k})\|}{\mu(\tilde{A})} \frac{\mu(\tilde{A})\alpha}{\varphi\|\tilde{x}(t_{k})\|} = \frac{\psi}{\varphi}\alpha,
\end{align*}
which completes the proof.
\end{proof}

In the following we will use Theorems~\ref{Theorem-7} and \ref{Theorem-vit} to derive an upper bound on the absolute value of approximation error for fuel savings. In our derivation, we use the following definitions 
\begin{align}
\mathcal{S}_{v,i,k} &\triangleq 3  \left(v_i(t_k)+\frac{\psi}{\varphi}\alpha \right)^2 \left(C_{\mathrm{B}}- C_{\mathrm{D},i}(d_i(t_k) - \alpha)\right), \label{eq:S_vk} \\ 
\mathcal{S}_{d,i,k} &\triangleq v^3_i(t_k) C_{\mathrm{D},i}'(d_i(t_k)-\alpha),\label{eq:S_dk}   
\end{align}
where $C'_{\mathrm{D},i}$ denotes the derivative of $C_{\mathrm{D},i}$, that is, 
\begin{align}
C_{\mathrm{D},i}'(y)\triangleq \left. \frac{\partial C_{\mathrm{D},i}(d)}{\partial d}\right |_{d=y},\quad y\in \mathbb{R}. \label{eq: CDprime}
\end{align}
Furthermore, we define $C''_{\mathrm{D},i}$ to be the second derivative of $C_{\mathrm{D},i}$, that is,
\begin{align}
C_{\mathrm{D},i}''(y)\triangleq \left. \frac{\partial^2 C_{\mathrm{D},i}(d)}{\partial d^2}\right |_{d=y},\quad y\in \mathbb{R}. \label{eq: CDprimepriime}
\end{align}

The following assumption specifies a minimum distance threshold $\underline{d}$ above which the aerodynamic drag coefficient $C_{\mathrm{D},i}$ satisfies inequalities with respect to its first two derivatives. 
\begin{assumption} \label{AssumptionCD}
There exists $\underline{d}\in (0,d_{\max})$ such that $C'_{\mathrm{D},i}(d)>0$  and $C''_{\mathrm{D},i}(d)<0$ hold for all $d \in [ \underline{d},d_{\max})$.  
\end{assumption}

The inequalities $C'_{\mathrm{D},i}(d)>0$ and $C''_{\mathrm{D},i}(d)<0$ in Assumption~\ref{AssumptionCD} guarantee that $C_{\mathrm{D},i}(d)>0$ increases with respect to $d$ with a diminishing rate. We note that the experimental results in \cite{hussein2021} can be used for determining numerical values for $\underline{d}$ and $d_{\max}$ in this assumption. 

Before we state the main result of this section, we provide two technical lemmas.


\begin{lemma} \label{lemmavelocity}
There exists $\hat{t}\in (t_k,t)$ such that 
\begin{align}
|v_i^3(t)-v_i^3(t_k)|=3v_i^2(\hat{t})|v_i(t) - v_i(t_k)|. \label{eq:v-result}
\end{align}
\end{lemma}
\begin{proof} The result holds directly if $v_i(t_k) = v_i(t)$, since both sides of \eqref{eq:v-result} would be $0$. Now, consider the case where $v_i(t_k) < v_i(t)$. Define $f\colon\mathbb {R}\to \mathbb{R}$ by $f(v)\triangleq v^3$. 
Then by mean value theorem, there exists $\hat{v}\in (v_i(t_k),v_i(t))$ such that  
    \begin{align}
&v_i^3(t)-v_i^3(t_k)=f(v_i(t))-f(v_i(t_k)) \nonumber \\
&\quad = f'(\hat{v})(v_i(t)-v_i(t_k))=3\hat{v}^2(v_i(t) - v_i(t_k)). \label{eq:f-v-result}
\end{align}
By taking the absolute value of the left side and the far right side of \eqref{eq:f-v-result}, and observing that $\hat{v} \geq 0$, we obtain 
    \begin{align}
|v_i^3(t)-v_i^3(t_k)|&=3\hat{v}^2|v_i(t) - v_i(t_k)|. \label{eq:hatv-result}
\end{align}
Now, by continuity of $v_i(\cdot )$, there exists $\hat{t}\in (t_k,t)$ such that $v_i(\hat{t})=\hat{v}$. By substituting $\hat{v}$ with $v_i(\hat{t})$, we obtain \eqref{eq:v-result}. 
The case where  $v_i(t) < v_i(t_k)$ follows from an argument similar to the one that we used for the case where $v_i(t_k) <v_i(t)$. 
\end{proof}
\begin{lemma}\label{lemmaCD}
Suppose Assumptions~\ref{Assumption-Conditions-of-Theorem} and \ref{AssumptionCD} are satisfied. If $\underline{d}+\alpha \leq d_i(t_k)<d_{\max} -\alpha$, then for every $t\in(t_k, t_{k+1})$, there exists $\hat{t}\in (t_k,t)$ such that 
\begin{align}
|C_{\mathrm{D},i} (d_i(t))-C_{\mathrm{D},i} (d_i(t_k))|=C'_{\mathrm{D},i}(d_i(\hat{t}))|d_i(t)-d_i(t_k)|. \label{eq:CD-result}
\end{align}
\end{lemma}
\begin{proof}
Since Assumption~\ref{Assumption-Conditions-of-Theorem} is satisfied, by Theorem~\ref{Theorem-7}, we have \eqref{eq:di-bound} for $t \in (t_k, t_{k+1})$. Now, since $\underline{d} + \alpha \le d_i(t_k) < d_{\mathrm{max}} - \alpha$, we have that both $d_i(t_k)$ and $d_i(t)$ are inside the interval $[\underline{d}, d_{\mathrm{max}}]$. By applying the mean value theorem, there exists $\hat{d}  \in (d_i(t_k), d_i(t))$ such that 
\begin{align}
C_{\mathrm{D},i}(d_i(t)) - C_{\mathrm{D},i}(d_i(t_k)) = C_{\mathrm{D},i}^{\prime}(\hat{d})(d_i(t) - d_i(t_k)). \label{eq: CDdi(t)-CDdi(tk)1}
\end{align}
Taking the absolute value of both sides yields
\begin{align}
|C_{\mathrm{D},i}(d_i(t)) - C_{\mathrm{D},i}(d_i(t_k))| = |C_{\mathrm{D},i}^{\prime}(\hat{d})||d_i(t) - d_i(t_k)|. \label{eq: CDdi(t)-CDdi(tk)2}  
\end{align}
Noting that $d_i(t_k),d_i(t)\in[\underline{d}, d_{\mathrm{max}}]$ and  $\hat{d}  \in (d_i(t_k), d_i(t))$, we have $\hat{d}\in(\underline{d},d_\mathrm{max})$. Therefore, it follows from Assumption~\ref{AssumptionCD} that $C_{\mathrm{D},i}^{\prime}(\hat{d}) > 0$. Thus, $| C_{\mathrm{D},i}^{\prime}(\hat{d})|= C_{\mathrm{D},i}^{\prime}(\hat{d})$ and 
\begin{align}
|C_{\mathrm{D},i}(d_i(t)) - C_{\mathrm{D},i}(d_i(t_k))| = C_{\mathrm{D},i}^{\prime}(\hat{d})|d_i(t) - d_i(t_k)|. \label{eq: CDdi(t)-CDdi(tk)3}   
\end{align}
Furthermore, by the continuity of $d_i(\cdot)$, there exists $\hat{t} \in (t_k, t)$ such that $\hat{d} = d_i(\hat{t})$. Substituting $d_i(\hat{t})$ for $\hat{d}$ yields \eqref{eq:CD-result}, which completes the proof.
\end{proof}

We are now ready to provide the main result of this section.

\begin{theorem} \label{thm:fuel_error_single_step}
Suppose Assumptions~\ref{Assumption-Conditions-of-Theorem} and \ref{AssumptionCD} are satisfied. Given $k\in \mathbb{N}_0$, consider the simulation interval $[t_k, t_{k+1})$. If $\underline{d}+\alpha \leq d_i(t_k)<d_{\max} -\alpha$, then the fuel savings approximation error $E_{i,k}$ defined in \eqref{eq:error_integral_def} is bounded as
\begin{equation} 
|E_{i,k}| \le (t_{k+1}-t_k) \frac{\rho K_\mathrm{vfa} \alpha}{2 Q_{\mathrm{LHV}} \eta_{\mathrm{engine}}} \left( \frac{\psi}{\varphi}\mathcal{S}_{v,i,k} + \mathcal{S}_{d,i,k} \right), \label{eq:fuel_bound_single}
\end{equation}
where $\mathcal{S}_{v,i,k}$ and $\mathcal{S}_{d,i,k}$ are given in \eqref{eq:S_vk} and \eqref{eq:S_dk}.
\end{theorem}
\begin{proof}
By \eqref{eq:Psave-expanded}, \eqref{eq:fuel_integral_tk-tk+1}, and \eqref{eq:error_integral_def}, we have 
\begin{align}
    E_{i,k} &=  \frac{1}{Q_{\mathrm{LHV}}\eta_\mathrm{engine}} \int_{t_k}^{t_{k+1}} P_{\mathrm{save},i}(\tau)-P_{\mathrm{save},i}(t_k)\mathrm{d}\tau \nonumber \\ 
    &=  \frac{\rho K_{\mathrm{vfa}}}{2Q_{\mathrm{LHV}}\eta_\mathrm{engine}} \int_{t_k}^{t_{k+1}}\left(C_{\mathrm{B}}-C_{\mathrm{D},i}(d_i(\tau))\right)v_i^3(\tau) \nonumber \\ 
    & \qquad-\left(C_{\mathrm{B}}-C_{\mathrm{D},i}(d_i(t_k))\right)v_i^3(t_k)\mathrm{d}\tau. \label{eq:E-integral-form_2}
    \end{align}
By taking the absolute value of both sides and using the triangle inequality,  we obtain
\begin{align}
|E_{i,k}| &\leq \frac{\rho K_\mathrm{vfa}}{2Q_\mathrm{LHV}\eta_\mathrm{engine}} \int_{t_k}^{t_{k+1}}\bigg|\left(C_{\mathrm{B}}-C_{\mathrm{D},i}(d_i(t))\right)v^3_i(t)\nonumber \\ 
&\quad \, -\left(C_{\mathrm{B}}-C_{\mathrm{D},i}(d_i(t_k))\right)v^3_i(t_k)\bigg| \mathrm{d}t. \label{eq:E_kstep1}
\end{align}
Inside the expression in absolute value on the right hand side of (\ref{eq:E_kstep1}), we add and subtract the same term $\left(C_{\mathrm{B}}-C_{\mathrm{D},i}(d_i(t))\right)v^3_i(t_k)$ to get
\begin{align}
|E_{i,k}| &\leq \frac{\rho K_\mathrm{vfa}}{2Q_\mathrm{LHV}\eta_\mathrm{engine}} \int_{t_k}^{t_{k+1}}\bigg|\left(C_{\mathrm{B}}-C_{\mathrm{D},i}(d_i(t))\right)v^3_i(t)\nonumber \\ 
&\quad \, -\left(C_{\mathrm{B}}-C_{\mathrm{D},i}(d_i(t))\right)v^3_i(t_k)\nonumber \\
&\quad \, +\left(C_{\mathrm{B}}-C_{\mathrm{D},i}(d_i(t))\right)v^3_i(t_k)\nonumber \\
&\quad \, -\left(C_{\mathrm{B}}-C_{\mathrm{D},i}(d_i(t_k))\right)v^3_i(t_k)\bigg| \mathrm{d}t \nonumber \\
&= \frac{\rho K_\mathrm{vfa}}{2Q_\mathrm{LHV}\eta_\mathrm{engine}}  \nonumber \\
&\quad \cdot \int_{t_k}^{t_{k+1}}\bigg|\big(C_{\mathrm{B}}-C_{\mathrm{D},i}(d_i(t))\big)\big(v^3_i(t)- v^3_i(t_k)\big)\nonumber \\
&\qquad \, + v_i^3(t_k)\big(C_{\mathrm{D},i}(d_i(t_k))-C_{\mathrm{D},i}(d_i(t))\big)\bigg| \mathrm{d}t \nonumber \\
&\leq \frac{\rho K_\mathrm{vfa}}{2Q_\mathrm{LHV}\eta_\mathrm{engine}}  \nonumber \\
&\quad \cdot \int_{t_k}^{t_{k+1}}\big(C_{\mathrm{B}}-C_{\mathrm{D},i}(d_i(t))\big) \big|v^3_i(t)- v^3_i(t_k)\big|\nonumber \\
&\qquad \, + v_i^3(t_k)|C_{\mathrm{D},i}(d_i(t))-C_{\mathrm{D},i}(d_i(t_k))| \mathrm{d}t,
\label{eq:E_kstep1.5}
\end{align}
where to obtain the last inequality, we also used the fact that $C_{\mathrm{B}}-C_{\mathrm{D},i}(d_i(t))\geq0$ and that the vehicles are always in forward motion, and hence, $v_i(t_k)\geq 0$.

Then for the terms $\big|v^3_i(t)- v^3_i(t_k)\big|$ and $|C_{\mathrm{D},i}(d_i(t))-C_{\mathrm{D},i}(d_i(t_k))|$ in (\ref{eq:E_kstep1.5}), we apply Lemmas~\ref{lemmavelocity} and \ref{lemmaCD}. It follows from those lemmas that for $t \in [t_k, t_{k+1})$, there exist $\hat{t}_1,\hat{t}_2\in (t_k,t)$ such that
\begin{align}
|E_{i,k}| &\leq \frac{\rho K_\mathrm{vfa}}{2Q_\mathrm{LHV}\eta_\mathrm{engine}} \nonumber \\ 
&\quad \cdot \int_{t_k}^{t_{k+1}} \bigg( (3v^2_i(\hat{t}_1)\big(C_{\mathrm{B}}-C_{\mathrm{D},i}(d_i(t))\big)\big|v_i(t)-v_i(t_k)\big| \nonumber \\ 
&\quad + v^3_i(t_k)C'_{\mathrm{D},i}(d_i(\hat{t}_2))\big|d_i(t)-d_i(t_k)\big| \bigg) \, \mathrm{d}t. 
\label{eq:E_k step2}
\end{align}
Now, by Theorem~\ref{Theorem-7} and Theorem~\ref{Theorem-vit}, the differences of inter-vehicle distance and velocity between $t$ and $t_k$ are bounded and the inequalities (\ref{eq:di-bound}) and (\ref{eq:vi-bound}) hold. A consequence of   (\ref{eq:vi-bound}) is that
\begin{align}
    |v_i(\hat{t}_1)| \le |v_i(t_k)| + \frac{\psi}{\varphi}\alpha, \label{velocity-bound}
\end{align} 
since $\hat{t}_1\in(t_k,t_{k+1})$. 
Furthermore, since Assumption~\ref{AssumptionCD} holds, $C_{\mathrm{D},i}(d)$ is an increasing function, and $C'_{\mathrm{D},i}(d)$ is a decreasing function of $d$. Therefore, considering $t\in [t_k,t_{k+1}), \hat{t}_2\in(t_k,t_{k+1})$, we have
 \begin{align}
     C_{\mathrm{B}}-C_{\mathrm{D},i}(d_i(t)) &\leq C_{\mathrm{B}} - C_{\mathrm{D},i}(d_i(t_k)-\alpha), \label{CD-bound-1} \\
 C'_{\mathrm{D},i}(d_i(\hat{t}_2))&\leq C'_{\mathrm{D},i}(d_i(t_k)-\alpha), \label{CD-bound-2}
  \end{align}
where we also used the fact that $d_i(t)\geq d_i(t_k)-\alpha$ and $d_i(\hat{t}_2)\geq d_i(t_k)-\alpha$, which are consequences of (\ref{eq:di-bound}).

 Substituting the bounds in \eqref{velocity-bound}--\eqref{CD-bound-2} into \eqref{eq:E_k step2} yields
\begin{align}
|E_{i,k}| &\leq \frac{\rho K_\mathrm{vfa} \alpha}{2Q_\mathrm{LHV}\eta_\mathrm{engine}} \nonumber \\ 
&\quad \cdot \int_{t_k}^{t_{k+1}}  3 \frac{\psi}{\varphi}\left(v_i(t_k)+\frac{\psi}{\varphi}\alpha \right)^2 \left(C_{\mathrm{B}}- C_{\mathrm{D},i}(d_i(t_k) - \alpha)\right) \nonumber \\ 
&\quad +  \left(v^3_i(t_k)C'_{\mathrm{D},i}(d_i(t_k) - \alpha)\right) \, \mathrm{d}t. \label{eq:E_k step3}
\end{align}
Next, we calculate the integral in (\ref{eq:E_k step3}) and substitute (\ref{eq:S_vk}) and (\ref{eq:S_dk}) to finally  obtain (\ref{eq:fuel_bound_single}), which completes the proof.
\end{proof}

The bound presented in Theorem~\ref{thm:fuel_error_single_step} is used in Section~\ref{sec:tradeoff} to find error bounds for simulation-based calculation of average fuel savings for a vehicle platoon.

\begin{remark}
For simplicity of the presentation, for safety and fuel efficiency analysis, we focused on the setting where the vehicles have the same characteristic time constant $\tau_{\mathrm{d}}$ and same control parameters $k_{\mathrm{p}}$ and $k_{\mathrm{d}}$. This setting can be extended by considering different time constants $\tau_{\mathrm{d},i}$ for $i\in\{0,1,\ldots,n\}$ and control parameters $k_{\mathrm{p},i}$ and $k_{\mathrm{d},i}$ for $i \in\{1,2,\ldots,n\}$. In this case, the matrices $A_{\mathrm{c}}$ and $B_{\mathrm{c}}$ defined in Section~\ref{sec:TheOverallDynamics} will have diagonal components $M$, $N_1,N_2,\ldots,N_n$, $E$ where $\tau_{\mathrm{d}}$ in $M$ and $E$ is replaced with $\tau_{\mathrm{d},0}$, and moreover, $\tau_{\mathrm{d}}$, $k_{\mathrm{p}}$, and $k_{\mathrm{d}}$  in $N_i$ are replaced with $\tau_{\mathrm{d},i}$, $k_{\mathrm{p},i}$ and $k_{\mathrm{d},i}$. With this new setting, Theorems~\ref{Theorem-Main}--\ref{thm:fuel_error_single_step} still hold true with the new $A_{\mathrm{c}}$ and $B_{\mathrm{c}}$ matrices, since the overall dynamics still follow the linear system \eqref{eq:linear-system}.
\end{remark}

\section{Simulation Studies on Safety and Fuel Efficiency}

\begin{figure}[t]
\centering 
\includegraphics[width=1\columnwidth]{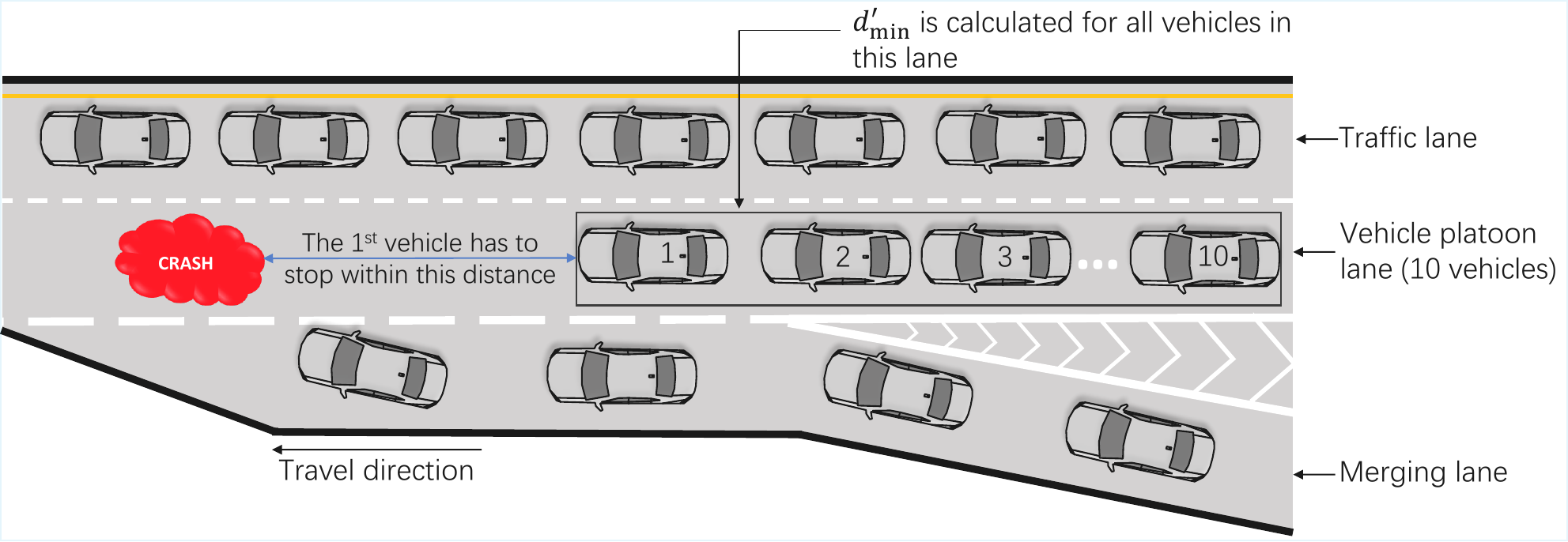}
\caption{Illustration of a sudden braking scenario on a highway under packet losses. The vehicle platoon is prevented from changing lanes due to congestion in the top lane and it has to make a sudden braking in response to a crash that was caused by merging ahead.}
 \label{Figure-new-illustration-vp}
\end{figure}

\label{sec:Numerical-Example}
In this section, we demonstrate our simulation approach by considering
the platoon dynamics described in Section~\ref{sec:Vehicle-Platoon-Dynamics}. We consider a platoon of $n=10$ vehicles shown in Figure~\ref{Figure-new-illustration-vp} with Predecessor-Following topology and Constant Time Gap spacing policy. In the scenario that we consider, there is a merging-related crash occurring in the lane of the platoon ahead of the leader vehicle. The top lane is congested with heavy traffic, and therefore, changing a lane is impossible for the vehicle platoon. Since there is a crash ahead, the platoon has to make a sudden brake. By applying our proposed simulation methods to this critical scenario, we aim to verify whether it is possible to maintain enough inter-vehicle distance to avoid a collision in the platoon. We identify a collision as the minimum inter-vehicle distance reaching the zero value. We also note that potential safety concerns may arise when this distance take small values close to zero even if does not reach it.

In our scenario, we set the characteristic constant as $\tau_{{\mathrm{d}}}=1.5$, time gap as
$h=0.6$, the length of all vehicles as $L_{i}=4.7\,{\mathrm{m}}$,
initial speed of all vehicles as $v_{i}(0)=30\,{\mathrm{m/s}}$, and
initial acceleration of all vehicles as $a_{i}(0)=0\,{\mathrm{m/s^{2}}}$.
Initial positions are set as $p_{i}(0)=200-(5+hv_{i}(0))i$ for $i\in\{1,2,\ldots,n\}$.
We consider the case where the vehicles communicate at every $T=0.1\,{\mathrm{s}}$.
We set the standstill distance as $r=5\,{\mathrm{m}}$. Notice that
the initial distance between the vehicles is shorter than $r+hv_{i}(0)$
by an amount equal to the length of the vehicles. A string-stable
platoon's control system is capable of handling this situation when
there are no packet losses and sudden brakes. 

In our simulations, we consider the case where a brake is suddenly introduced after
time $t_{{\mathrm{brake}}}=5\,{\mathrm{s}}$ in the form of (\ref{eq:discrete-brake-model-1-revised}), \eqref{eq:discrete-brake-model-2-revised}
where $\gamma=1.2$ and $\eta=0.1$. The brake is introduced by the leader vehicle to avoid a crash ahead and it is applied until the
simulation end time $t_{{\mathrm{end}}}=25\,{\mathrm{s}}$. Along with
the sudden brake, we consider the case where communication between
the vehicles faces packet losses. In particular, we consider two different
models for packet losses: (1) consecutive packet losses and (2) random
packet losses. 

For each packet loss model, we attempt to approximate $d_{\min}^{*}$ (the minimum
of minimum inter-vehicle distances) by using the distance data $d_{i}(t_{k})$
obtained through our discrete-time simulation approach. Specifically,
we use this data to calculate ${d_{\min}}$ according to (\ref{eq:dmin-discrete}).
Notice that ${d_{\min}}$ is only an approximation of the true
value $d_{\min}^{\star}$ given in (\ref{eq:d-min-def}). However,
our simulation guarantees that $|d_{\min}^{\star}-d_{\min}|\leq\alpha$.
We follow the approach in Section~\ref{subsection:design of simulation times} to choose the
consecutive times $t_{k}$, $k\in{\mathbb{N}}_{0}$, where we evaluate
the state of the vehicle platoon dynamics. As both Theorems~\ref{Theorem-Main} and \ref{Theorem-7}
imply $|d_{i}(t)-d_{i}(t_{k})|\leq\alpha$, $t\in[t_{k},t_{k+1})$,
our approach guarantees that $|d_{\min}^{\star}-d_{\min}|\leq\alpha$ in both cases. In our simulations we choose $\alpha=1$.
As a result, observing $d_{\min}>1$ indicates that a collision between
vehicles is definitely avoided. However, a value of $d_{\min}\in(0,1)$
does not provide any safety guarantees. To further enhance the efficiency in our numerical analysis, we stop simulations when a collision occurs or vehicles come to a stop. As a quantitative result from our simulations, we obtain
\begin{align}
d^{\prime}_{\min} \triangleq \min_{i\in\{2,\ldots,n\}}\min_{k\in\{0,\ldots,k^{\prime}_{{\mathrm{end}}}\}}d_{i}(t_{k}),
\end{align}
where $k^{\prime}_{\mathrm{end}}$ denotes the total simulation step count for each simulation, which is determined either when the simulation time reaches $t_\mathrm{end}$, or when vehicle velocities reach $0$ or a collision occurs.

\subsection{$\ell$-consecutive packet losses} 
\label{sec:ell-consecutive-Numerical}

\begin{figure}[t]
\centering 
\includegraphics[width=0.9\columnwidth]{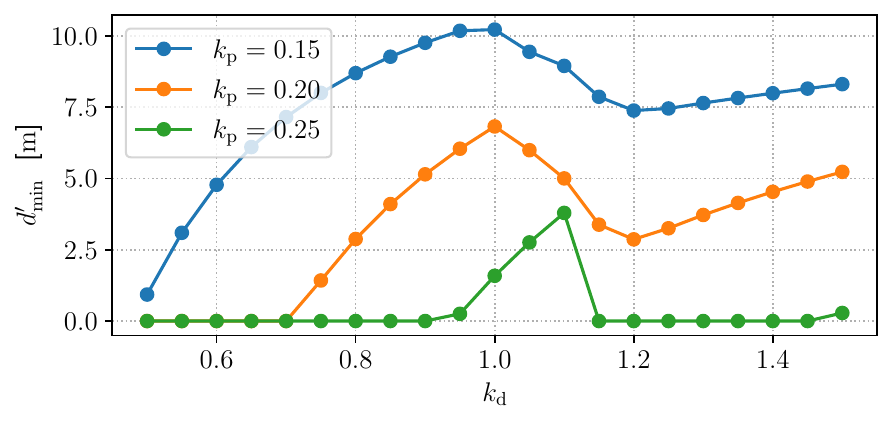}
\caption{Minimum of minimum inter-vehicle distances ($d^{\prime}_{\min}$) of Theorem~\ref{Theorem-7} with respect
to control parameters ($k_{{\mathrm{p}}}$ and $k_{{\mathrm{d}}}$).}
 \label{Figure-dmin-kp-kd}
\end{figure}

\begin{figure}[t]
\centering 
\includegraphics[width=0.9\columnwidth]{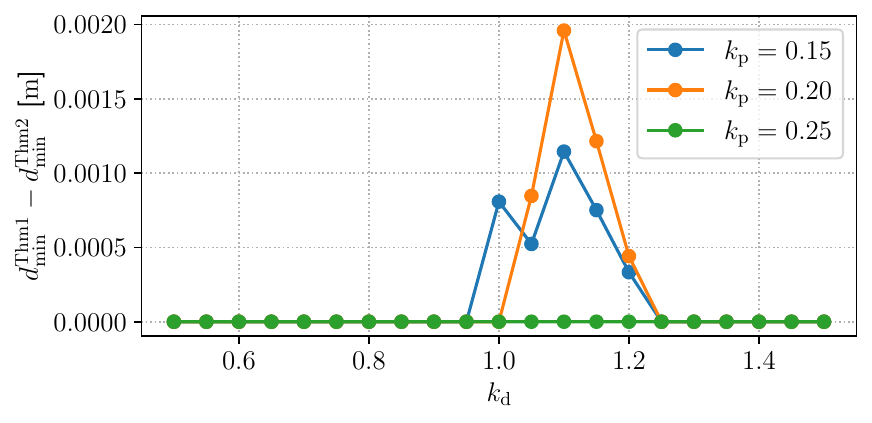}
\label{fig:dmin_thm2}
\caption{Difference of $d^{\prime}_{\min}$ obtained by using the simulation approaches in Theorem~\ref{Theorem-Main} and in Theorem~\ref{Theorem-7} with respect to control parameters ($k_{{\mathrm{p}}}$ and $k_{{\mathrm{d}}}$).}
 \label{Figure-dmin-difference}
\end{figure}

First, we consider the scenario where there are $\ell=7$ packet losses after every successful communication attempt. For this scenario Figure~\ref{Figure-dmin-kp-kd} shows how $d^{\prime}_{\min}$ changes with respect to control parameters $k_{{\mathrm{p}}}$
and $k_{{\mathrm{d}}}$ for Theorem~\ref{Theorem-7}. The values obtained with Theorem~\ref{Theorem-Main} are similar and they differ from the values of Theorem~\ref{Theorem-7} at most by $0.002\,\mathrm{m}$ as shown in Figure~\ref{Figure-dmin-difference}. We observe in most cases that for large values of $k_{{\mathrm{d}}}$
and small values of $k_{{\mathrm{p}}}$, the value of $d'_{\min}$ is
larger. This is consistent with the string stability framework developed
in \cite{dolk2017}. In that work, the condition $k_{{\mathrm{d}}}>\tau_\mathrm{d} k_{{\mathrm{p}}}$
was identified as a condition for string stability. We also note that $d^{\prime}_{\min}=0$ 
for the cases where $k_\mathrm{p} = 0.15$, $k_\mathrm{d}=0.5$ $k_{{\mathrm{p}}}=0.2$, $k_{{\mathrm{d}}} \leq 0.7$ and  $k_{{\mathrm{p}}}=0.25$, $k_{{\mathrm{d}}} \leq 0.95$ or $ k_{{\mathrm{d}}} \geq 1.15$. These indicate a collision among vehicles. Additionally, Figure~\ref{Figure-steps} shows that Theorem~\ref{Theorem-7} requires approximately ten times fewer total simulation steps ($k'_{\mathrm{end}}$) compared to Theorem~\ref{Theorem-Main}. This indicates the advantage of Theorem~\ref{Theorem-7}. In Figure~\ref{Figure-steps}, when $k_{\mathrm{p}}$ is fixed, an increase in $k_{\mathrm{d}}$ occasionally results in a large increase in $k^{\prime}_{\mathrm{end}}$ values. This is because the simulation runs for longer time (without being stopped) so that it results in larger numbers of total simulation steps $k^{\prime}_\mathrm{end}$. We also compare our simulation approaches with a constant step baseline, which sets $\nu_k=1$ in \eqref{eq:tset-k} that results in the fixed simulation interval
$t_{k+1}-t_k=T/\overline{N}=5\times10^{-6}\,\mathrm{s}$. This particular setting provides the same approximation error guarantees as our approaches. However, as Figure~\ref{fig:baseline_2thms_steps} shows, the baseline is much slower. Our first simulation approach via Theorem~\ref{Theorem-Main} requires approximately ten times fewer simulation steps compared to the constant step baseline, while Theorem~\ref{Theorem-7} requires approximately ten times fewer steps than Theorem~\ref{Theorem-Main}.

\begin{figure}[t]
\centering 
\subfloat{
\includegraphics[width=0.9\linewidth]{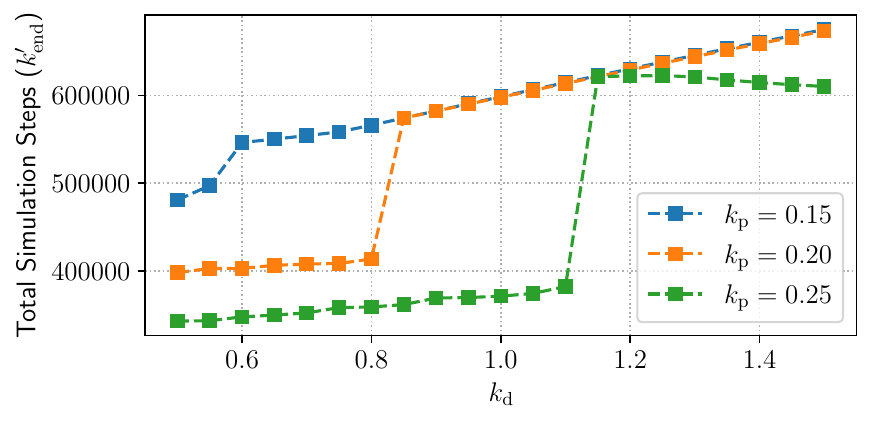}
\label{fig:steps_thm1} 
} \\ 
\subfloat{
\includegraphics[width=0.9\linewidth]{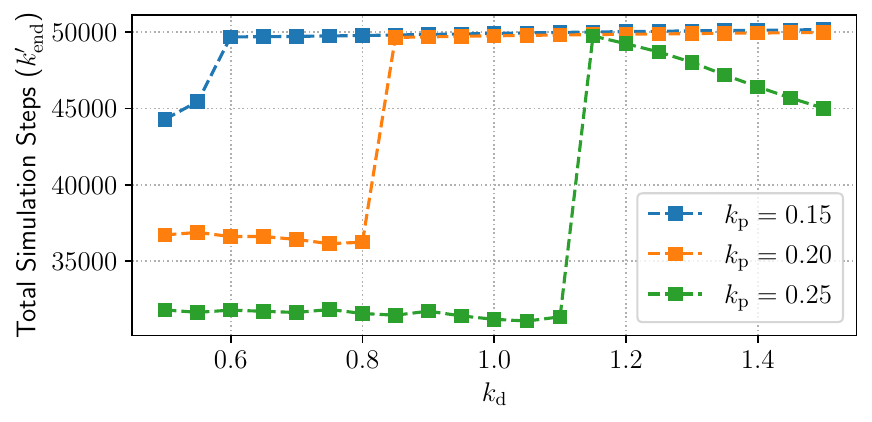}
\label{fig:steps_thm2}
}
\caption{Plots of total simulation steps $k'_{\mathrm{end}}$ with respect to control parameters $k_{{\mathrm{p}}}$ and $k_{{\mathrm{d}}}$ for Theorem~\ref{Theorem-Main} (top) and Theorem~\ref{Theorem-7} (bottom).}
 \label{Figure-steps}
\end{figure}

\begin{figure}[t]
\centering 
\includegraphics[width=0.9\columnwidth]{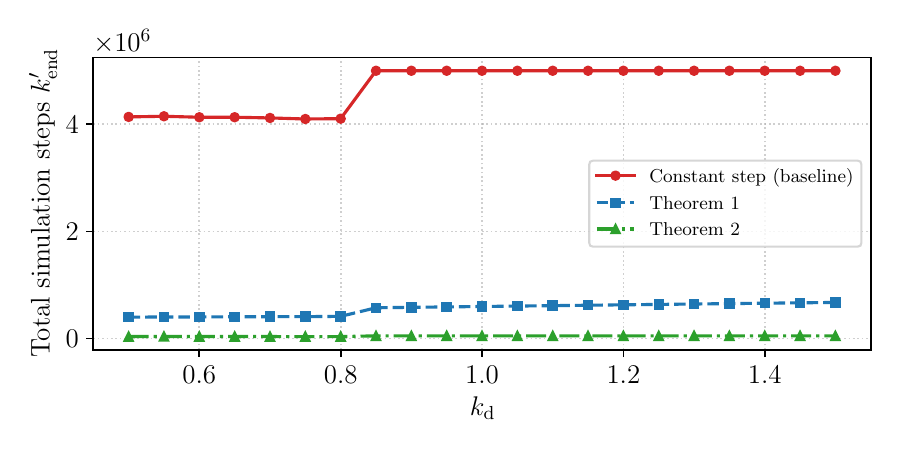}
\caption{The comparison of total simulation steps for a constant step baseline, Theorem~\ref{Theorem-Main}, and Theorem~\ref{Theorem-7} for $k_{\mathrm{p}} = 0.2$ and varying $k_{\mathrm{d}}$.}
 \label{fig:baseline_2thms_steps}
\end{figure}

\subsection{Random packet losses} \label{Random-packet-losses}
Second, we consider random packet losses with a relatively large packet loss probability
$p=0.8$. Simulations under such a large packet loss probability are useful to understand the vehicle platoon's safety in worst-case scenarios that may even be caused intentionally by malicious actors. For the particular case of $p=0.8$, we run 10000 simulations by using the simulation approach described in Theorem~\ref{Theorem-7} for two settings. In the first setting, we have $k_{{\mathrm{p}}}=0.2$ and $k_{{\mathrm{d}}}=1.0$. In the second setting, we have $k_{{\mathrm{p}}}=0.25$ and $k_{{\mathrm{d}}}=1.0$. On a standard laptop with 16 GB Memory and Ryzen 7 5800H CPU, simulations for these two settings took around 30 hours in total. The histograms of $d'_{\min}$ are provided
in Figure~\ref{Figure-histogram}. We observe that for the first setting ($k_{{\mathrm{p}}}=0.2$ and $k_{{\mathrm{d}}}=1.0$), in some simulations
the packet loss pattern realizations result in safety issues, while in most of the realizations the inter-vehicle distances stay sufficiently large. In the second setting ($k_{{\mathrm{p}}}=0.25$ and $k_{{\mathrm{d}}}=1.0$), the obtained inter-vehicle distances are smaller and there are more critical safety issues. In the first setting with $k_{\mathrm{p}}=0.2,k_{\mathrm{d}}=1.0$, having a large number of simulations is important to see potential issues, because as the histogram indicates, in a typical simulation $d'_{\min}$ appears to be large enough (larger than $8$). However, although rarely, there are packet loss realizations that result in a collision. 

\begin{figure}[t]
\centering
\subfloat{
\includegraphics[width=0.9\columnwidth]{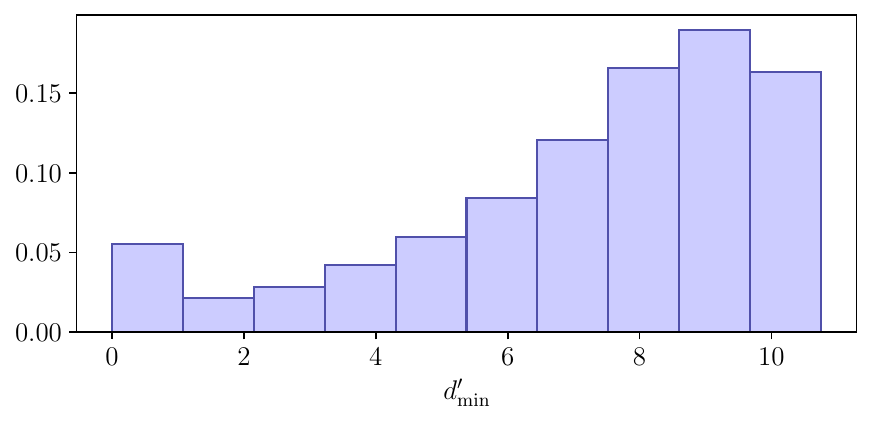}
\label{fig:histogram_Theorem7_kp_0.2_kd_1.2.pdf}
} \\
\subfloat{ 
\includegraphics[width=0.9\columnwidth]{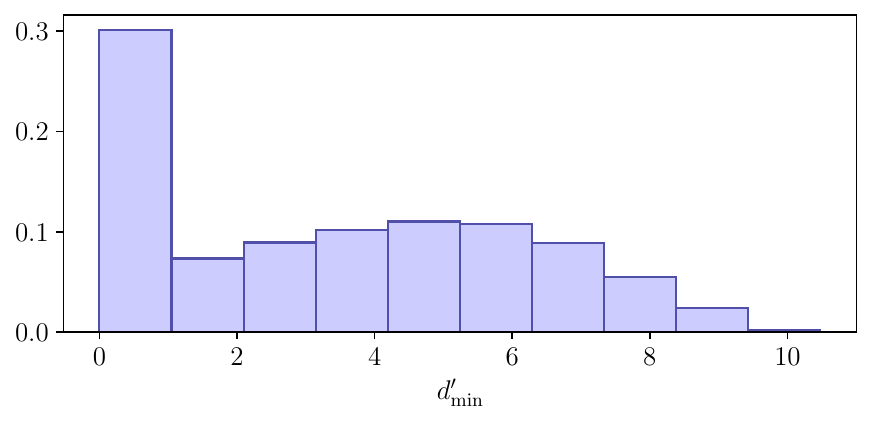}
\label{fig:histogram_Theorem7_kp_0.25_kd_1.2.pdf}
}
\caption{Plot of the minimum of minimum inter-vehicle distances ($d'_{\min}$)
over 10000 simulations for Theorem~\ref{Theorem-7} under random packet losses with respect to control parameters $k_\mathrm{p}$ and $k_\mathrm{d}$. (Top: $k_{\mathrm{p}}=0.2,k_{\mathrm{d}}=1.0$, Bottom: $k_{\mathrm{p}}=0.25,k_{\mathrm{d}}=1.0$) } 
 \label{Figure-histogram}
\end{figure}

\begin{figure}[t]
\centering 
\includegraphics[width=0.9\columnwidth]{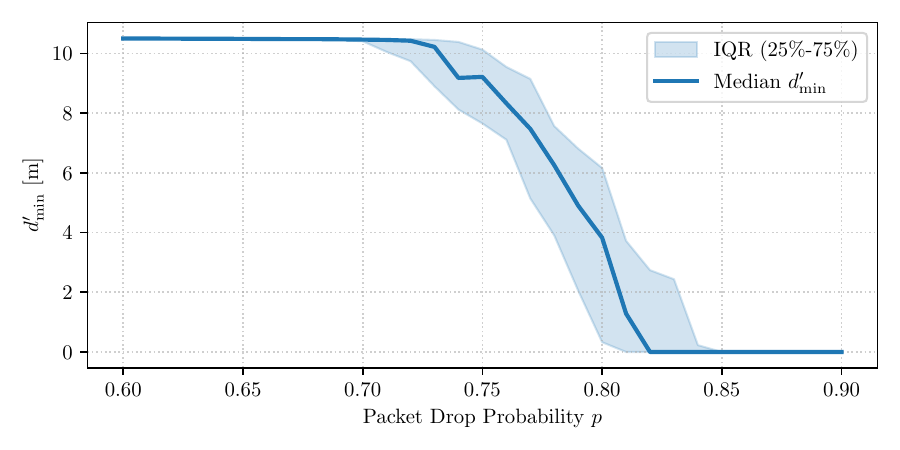}
\caption{Plot of the minimum of minimum inter-vehicle distances ($d'_{\min}$) with respect to the packet drop probability $p \in [0.6, 0.9]$. For each $p$, 100 simulations are conducted using the method in Theorem~\ref{Theorem-7}.}
 \label{Figure-histogram-different-p}
\end{figure}

To further evaluate $d'_{\min}$ under random packet losses, we run extensive simulations by varying packet loss probability $p$ from $0.6$ to $0.9$ with an increase of $0.01$ for control parameters values $k_p=0.2$ and $k_\mathrm{d}=0.7$. For each value of $p$, we run $100$ simulations. The results are shown in Figure~\ref{Figure-histogram-different-p}, where the solid line represents the median value over all 100 simulations, while the shaded area represents the Interquartile Range (IQR), spanning from the 25th to the 75th percentile. This figure illustrates the resulting trend of $d'_{\min}$ with respect to the packet loss probability $p$. In particular, we notice that $d'_{\min}$ decreases sharply as $p$ exceeds $0.7$ and its lower quartile converges to 0 when $p = 0.82$. The result indicates a sudden decline in safety over a relatively small increase in packet loss probability.

\begin{figure}[t]
\centering 
\includegraphics[width=1\columnwidth]{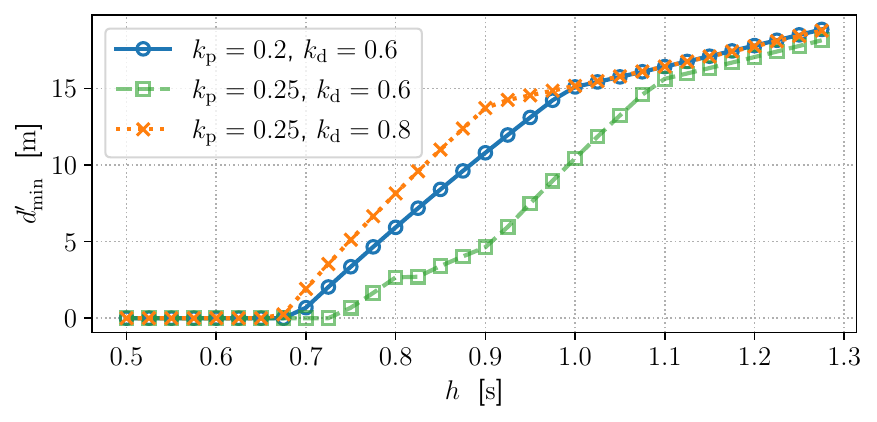}
\caption{Plot of $d'_\mathrm{min}$ with respect to time-gap $h$ for different control parameters ($k_\mathrm{p}$ and $k_\mathrm{d}$) under $l$-consecutive losses.}
 \label{d_prime_min_vs_h}
\end{figure}

\subsection{Effect of time-gap $h$ on safety}
Next, to investigate how the time-gap $h$ affects safety, we run extensive simulations using our simulation framework characterized by Theorem~\ref{Theorem-7}.

In particular, we consider $\ell$-consecutive packet losses with $\ell=7$ as described in Section~\ref{sec:ell-consecutive-Numerical} and simulate the vehicle platoon for different values of $h$ and different values of control parameters $k_{\mathrm{p}}$ and $k_{\mathrm{d}}$. Our simulation approach allows us to find critical values of $h$ for which the platoon faces a collision. 

We obtain Figure~\ref{d_prime_min_vs_h} which shows $d^{\prime}_{\mathrm{min}}$ with respect to $h$. We note that when the time gap is small such that $h <  0.675$, we have $d^{\prime}_{\mathrm{min}}=0$ for control parameters $k_{\mathrm{p}}=0.25, k_{\mathrm{d}}=0.8$ (see orange line) indicating a crash. In addition, when $h <  0.7$, we obtain $d^{\prime}_{\mathrm{min}}=0$ indicating crash for $k_{\mathrm{p}}=0.2, k_{\mathrm{d}}=0.6$ (see blue line), and when $h<0.75$, we also obtain $d^{\prime}_{\mathrm{min}}=0$ for $k_{\mathrm{p}}=0.25, k_{\mathrm{d}}=0.6$ (see green line). We notice that $d^{\prime}_{\mathrm{min}}$ increases when $h$ increases, as expected. However as seen in Figure~\ref{d_prime_min_vs_h}, the relationship is not linear. Our observations regarding the time gap are consistent with those  of~\cite{RAJAMANI200115} in that small values of $h$ are detrimental to vehicle platoon operation. This hints the possibility of achieving better outcomes using different spacing policies. To further highlight this point, in Section~\ref{sec:tradeoff}, we compare Constant Time Gap and Constant Distance Gap policies in terms of safety and fuel efficiency.


\subsection{Effects of vehicle and road parameters on safety and stopping distance}
\label{sec}

The discussion above focuses on the control parameters $k_\mathrm{p}$ and $k_\mathrm{d}$ and the time gap $h$. We next investigate how the braking parameter $\gamma$, packet-loss probability $p$, characteristic time constant $\tau_\mathrm{d}$, and initial velocity $v_0(0)$ affect the minimum inter-vehicle distance and stopping distance. For this investigation, we use our simulation approach based on Theorem~\ref{Theorem-7}.

\begin{figure}[t]
\centering
\includegraphics[width=0.9\columnwidth]{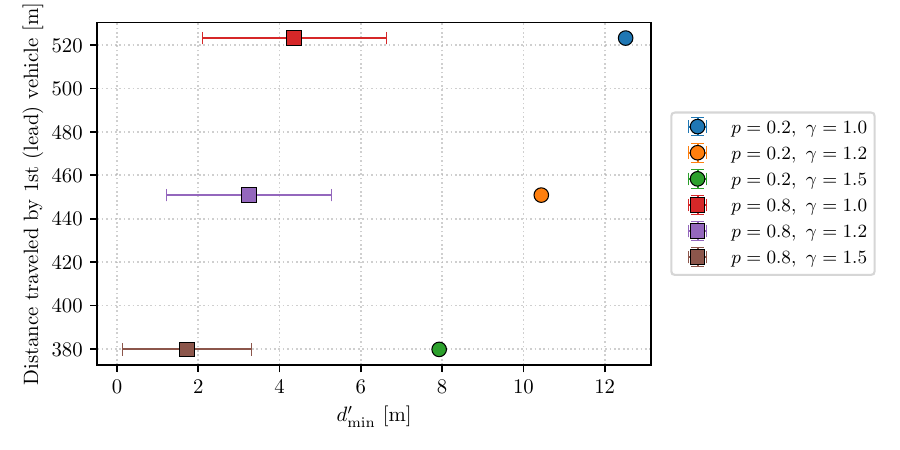}
\caption{Distance traveled by the 1st (lead) vehicle versus $d'_\mathrm{min}$. Each marker represents the mean over ten independent packet-loss realizations. The horizontal error bars indicate one standard deviation of $d'_\mathrm{min}$.}
\label{distance_vs_dmin_fixed_eta}
\end{figure}

\begin{figure}[t]
\centering
\includegraphics[width=0.9\columnwidth]{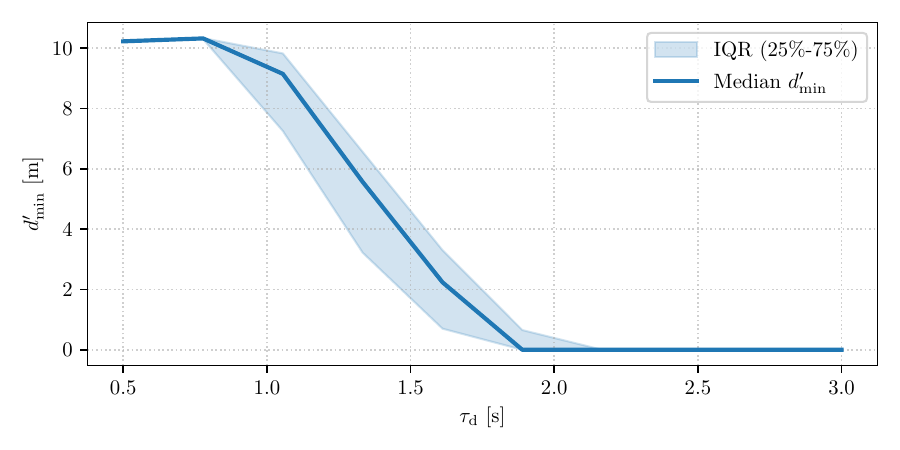}
\caption{Median $d'_\mathrm{min}$ as a function of $\tau_\mathrm{d}$ under random packet losses with $p=0.8$. The shaded region represents the interquartile range (IQR).}
\label{dmin_vs_taud_iqr}
\end{figure}

First, we investigate how the deceleration parameter affects the minimum inter-vehicle distance. In particular, for $\eta=0.15$, Figure~\ref{distance_vs_dmin_fixed_eta} shows the distance traveled by the first (lead) vehicle, as well as the mean and standard deviation of the minimum inter-vehicle distances obtained from ten different packet-loss realizations, for different values of the packet-loss probability $p$ and deceleration rate $\gamma$. We observe that increasing $\gamma$ decreases both the distance traveled by the lead vehicle and the minimum inter-vehicle distance $d'_{\mathrm{min}}$. For a fixed value of $\gamma$, increasing $p$ further reduces $d'_{\mathrm{min}}$. For a high packet-loss probability of $p=0.8$, the critical deceleration rate is approximately $\gamma=1.5$. Above this value, the minimum inter-vehicle distance is highly likely to reach $0$, indicating a collision.

Second, we investigate how the time constant $\tau_\mathrm{d}$ affects safety. This parameter characterizes the response time of the vehicles on given road conditions that the platoon travels. With the packet-loss probability fixed at $p=0.8$, Figure~\ref{dmin_vs_taud_iqr} shows the median minimum inter-vehicle distance $d'_{\mathrm{min}}$ over 100 runs with different packet-loss realizations as a function of $\tau_\mathrm{d}$. In particular, we consider ten equally spaced values of $\tau_\mathrm{d}$ between $0.5\,\mathrm{s}$ and $3.0\,\mathrm{s}$. We observe that $d'_\mathrm{min}$ decreases as $\tau_{\mathrm{d}}$ increases and that the median value of $d'_{\mathrm{min}}$ reaches $0$ at approximately $\tau_{\mathrm{d}}=1.9\,\mathrm{s}$.

Furthermore, we investigate the effect of the initial velocity $v_0(0)$ on platoon safety. Specifically, we fix the packet-loss probability at $p=0.8$ and the braking parameters at $\gamma=1.5$ and $\eta=1/6$. Figure~\ref{dmin_and_obstacle_distance_vs_v0} shows the median minimum inter-vehicle distance for different initial velocities over 100 runs with different packet-loss realizations. In particular, we consider ten equally spaced initial velocities between $25\,\mathrm{m/s}$ and $35\,\mathrm{m/s}$. For each case, we fix the initial inter-vehicle distance at $18.3\,\mathrm{m}$ and consider an obstacle, or crash site, located $380\,\mathrm{m}$ ahead of the first (lead) vehicle when braking begins, as illustrated in Figure~\ref{Figure-new-illustration-vp}.

In Figure~\ref{dmin_and_obstacle_distance_vs_v0}, $d_\mathrm{obs}$ denotes the remaining distance between the first vehicle and the obstacle when the vehicle comes to a stop. If $d_\mathrm{obs}$ reaches $0$, the first vehicle hits the obstacle, or crash site, represented by the red cloud shape in Figure~\ref{Figure-new-illustration-vp}. Figure~\ref{dmin_and_obstacle_distance_vs_v0} shows that $d_\mathrm{obs}$ reaches $0$ when $v_0(0)>30.5\,\mathrm{m/s}$, whereas the median value of $d'_\mathrm{min}$ reaches $0$ when $v_0(0)>31.5\,\mathrm{m/s}$. Both cases result in traffic accidents involving the vehicle platoon.

\begin{figure}[t]
\centering 
\includegraphics[width=1\columnwidth]{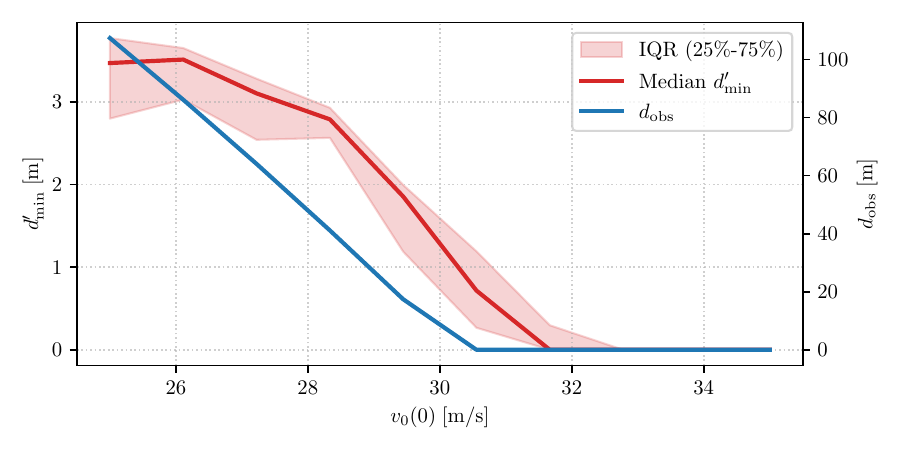}
\caption{Plot of median $d'_\mathrm{min}$ and $d_\mathrm{obs}$ with respect to $v_0(0)$ under random packet losses with $p=0.8$.}
\label{dmin_and_obstacle_distance_vs_v0}
\end{figure}

\subsection{Fuel efficiency analysis and a trade-off with safety}
\label{sec:tradeoff}
Next, we investigate the fuel efficiency of the vehicle platoon. For our simulations, the parameters of the aerodynamic drag model and fuel consumption are set according to the empirical data for light-duty vehicles in \cite{hussein2021}. In particular, the baseline drag coefficient for a single vehicle is $C_{\mathrm{B}}= 0.367$, the air density is $\rho = 1.2\,\mathrm{kg/m^3}$, the vehicle frontal area is $K_{\mathrm{vfa}} = 2.3\,\mathrm{m^2}$, the engine efficiency is $\eta_\mathrm{engine} = 0.3$, and the fuel lower heating value is $Q_{\mathrm{LHV}} = 43000\,\mathrm{J/g}$. Additionally, the coefficients $\hat{w}_{i,j}$, $\tilde{w}_{i,j}$, for $j\in\{0,1,2,3\}$, $\underline{d}_i$, and $d_{\mathrm{max},i}$ in \eqref{eq:rational_poly} depend on the vehicle's specific position in the platoon as shown in Table~\ref{table:dragparameter}. We note that we apply the lead vehicle drag model of \cite{hussein2021} to $2$nd vehicle in the platoon, and our notation for vehicle indices differs from theirs in that in our notation, $d_i(t)$ is the inter-vehicle distance between $(i-1)$th and $i$th vehicles. To determine the values of $\underline{d}_i$ for Assumption~\ref{AssumptionCD}, we simulate the numerical values of aerodynamic drag model. In particular, we obtain the values of $\underline{d}_i$ by numerically computing $C'_{\mathrm{D},i}(d)$ and $C''_{\mathrm{D},i}(d)$ for 100000 consecutive values of $d$ between $0.01\,\mathrm{m}$ and $100\,\mathrm{m}$, which allows us to find ranges of distance values that satisfy the conditions in Assumption~\ref{AssumptionCD}. 

\begin{figure}[t]
\centering 
\includegraphics[width=1\columnwidth]{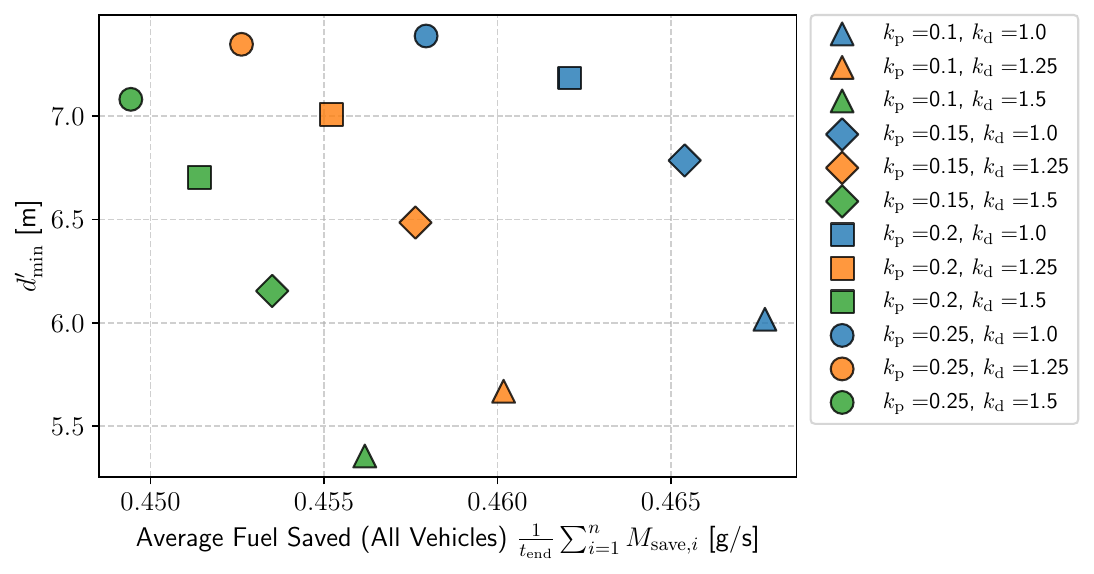}
\caption{Plot of minimum inter-vehicle distance $d'_\mathrm{min}$ versus average fuel saved for different control parameters ($k_\mathrm{p}, k_\mathrm{d}$) under Constant Time Gap policy. For fixed $k_\mathrm{p}$, smaller $k_\mathrm{d}$ enhances fuel economy without safety loss. However, for fixed $k_\mathrm{d}$, there is a trade-off in that larger $k_\mathrm{p}$ values achieve safer but less economic results.}
\label{SummaryFuelvsSafety}
\end{figure}

\begin{table*}[]
    \centering
    \caption{Aerodynamic drag model parameters assigned to each vehicle}
    \label{tab:vehicle_parameters}
    \resizebox{\textwidth}{!}{
    \begin{tabular}{l cccc cccc c c}
        \toprule
        $\boldsymbol{i}$
        & $\boldsymbol{\hat{w}_{i,3}}$ & $\boldsymbol{\hat{w}_{i,2}}$ & $\boldsymbol{\hat{w}_{i,1}}$ & $\boldsymbol{\hat{w}_{i,0}}$ 
        & $\boldsymbol{\tilde{w}_{i,3}}$ & $\boldsymbol{\tilde{w}_{i,2}}$ & $\boldsymbol{\tilde{w}_{i,1}}$ 
        & $\boldsymbol{\tilde{w}_{i,0}}$ &
       $\boldsymbol{\underline{d}_i}$ &
        $\boldsymbol{d_{\mathrm{max},i}}$  \\
        \midrule
        \midrule
        $2$
        & $0.14725$ & $-0.21819$ & $-0.091455$ & $0.58344$ 
        & $0.14192$ & $-0.14349$ & $-0.28946$ & $0.97713$ & $1.932$  & $80.000$ \\
        $3$    
        & $0.08770$ & $-0.39570$ & $-0.11120$ & $1.75980$ 
        & $0.08380$ & $-0.15700$ & $-1.54380$ & $4.17810$ & $3.510$ & $56.362$ \\
        4$\sim$10
        & $0.14980$ & $1.17190$ & $-3.27850$ & $2.44940$ 
        & $0.13800$ & $2.15650$ & $-5.48820$ & $3.93920$ & $2.733$ & $80.639$ \\
        \bottomrule
    \end{tabular} \label{table:dragparameter}
    }
\end{table*}

To evaluate the fuel efficiency of the vehicle platoon and investigate its relationship with safety, we introduce a continuous velocity oscillation scenario for the vehicle platoon with $n=10$ vehicles as described in the beginning of this section. We set the total simulation time to be $300\,\mathrm{s}$, that is  $t_{\mathrm{end}}=300$, and the virtual reference vehicle to repeatedly accelerate and decelerate. At the first $5$ seconds, $v_0(t)$ is constant, the virtual reference starts to decelerate after $t=5\,\mathrm{s}$. The desired acceleration $u_0(t)$ is set to oscillate between a maximum acceleration of $\gamma = 5.0\,\mathrm{m/s^2}$ and a maximum deceleration of $\gamma = -5.0\,\mathrm{m/s^2}$. The maximum deceleration is triggered when $v_0(t)$ reaches $30\,\mathrm{m/s}$, and the maximum acceleration is triggered when $v_0(t)$ reaches $5\,\mathrm{m/s}$. So it allows us to evaluate the fuel savings over time alongside the minimum inter-vehicle distances. For our simulation we consider $\ell$-consecutive packet losses where $\ell = 1$, that is, half of the packets are unable to be delivered.

Figure~\ref{SummaryFuelvsSafety} illustrates the minimum inter-vehicle distances $d'_\mathrm{min}$ versus the average fuel saved for all vehicles in the vehicle platoon for various combinations of 2 control parameters $k_\mathrm{p}$ and $k_\mathrm{d}$. In Figure~\ref{SummaryFuelvsSafety}, markers of the same color represent the same value of $k_\mathrm{d}$, and markers of the same shape represent the same value of $k_\mathrm{p}$. As observed in Figure~\ref{SummaryFuelvsSafety}, there are 3 distinct trade-off lines between safety and fuel efficiency, mainly influenced by the control parameter $k_\mathrm{p}$. Specifically, a larger $k_\mathrm{p}$ with the same value of $k_\mathrm{d}$ leads the inter-vehicle distances $d'_\mathrm{min}$ to be greater so that the vehicle platoon becomes safer, while it decreases the value of average fuel saved. On the other hand, Figure~\ref{SummaryFuelvsSafety} also reveals that for this particular scenario, decreasing the value of $k_\mathrm{d}$ while keeping $k_\mathrm{p}$ constant can enhance fuel efficiency without compromising the safety of the vehicle platoon. However, we cannot forever decrease the value of $k_\mathrm{d}$, as it may violate the string stability condition $k_\mathrm{d}>\tau_\mathrm{d}k_\mathrm{p}$.

We provide numerical values of average fuel savings in Table~\ref{table1}. There, we also provide error bounds. These error bounds are obtained by using the right-hand side of the inequalities provided in (\ref{eq:fuel_bound_single}) of Theorem~\ref{thm:fuel_error_single_step}. We observe that the bound is around the level of $15\%$. This is consistent across all parameter values, indicating that the largest simulation error is guaranteed to be below the $15\%$ of the numerical value obtained.

\begin{table}[]
\centering
\renewcommand{\arraystretch}{1.3}
\caption{Fuel Saving Simulation Results}
\begin{tabular}{llcc}
\toprule
$\boldsymbol{k_{\mathrm{p}}}$ & $\boldsymbol{k_{\mathrm{d}}}$
& \textbf{Avg. Fuel Savings [g/s]}
& \textbf{Error Bound [g/s]} \\
\midrule \midrule
0.1 & 1.0 & 0.4677 & 0.06846 \\
0.1 & 1.25 & 0.4602 & 0.06802 \\
0.1 & 1.5 & 0.4562 & 0.06783 \\
0.15 & 1.0 & 0.4654 & 0.06827 \\
0.15 & 1.25 & 0.4576 & 0.06775 \\
0.15 & 1.5 & 0.4535 & 0.06750 \\
0.2 & 1.0 & 0.4621 & 0.06807 \\
0.2 & 1.25 & 0.4552 & 0.06756 \\
0.2 & 1.5 & 0.4514 & 0.06730 \\
0.25 & 1.0 & 0.4579 & 0.06781 \\
0.25 & 1.25 & 0.4526 & 0.06737 \\
0.25 & 1.5 & 0.4494 & 0.06714 \\ 
\bottomrule
\end{tabular} \label{table1}
\end{table}

Through additional simulations, where we set the desired acceleration $u_0(t)$ to oscillate between a maximum acceleration of $\gamma = 0.5\,\mathrm{m/s^2}$ and a maximum deceleration of $\gamma = -0.5\,\mathrm{m/s^2}$, we observe that considerably smaller values of $k_{\mathrm{p}}$ and $k_{\mathrm{d}}$ (e.g., $k_{\mathrm{d}}\in\{0.15,0.2,0.25\}$ and $k_{\mathrm{p}}\in\{0.005,0.0075,0.01,0.0125\}$) can yield substantially larger values of $d'_{\min}$. However, because of the aforementioned tradeoff, these parameter values result in lower fuel savings, as shown in Figure~\ref{SummaryFuelvsSafety_CTGPF}.

Adopting a different spacing policy appears to improve both the minimum inter-vehicle distance and fuel efficiency. Figure~\ref{SummaryFuelvsSafety_CDGPF} shows the minimum inter-vehicle distance versus the average fuel savings under the Constant Distance Gap (CDG) policy. We observe that both quantities are higher than those obtained under the Constant Time Gap policy (CTG) shown in Figure~\ref{SummaryFuelvsSafety_CTGPF}. In addition, under the CDG policy, $k_{\mathrm{d}}$ becomes the key parameter governing the position of the results along the tradeoff curve. A larger value of $k_{\mathrm{d}}$ increases the minimum inter-vehicle distance while reducing the average fuel savings, whereas a smaller value decreases the minimum inter-vehicle distance while increasing the average fuel savings.

\begin{figure}[t]
\centering 
\includegraphics[width=1\columnwidth]{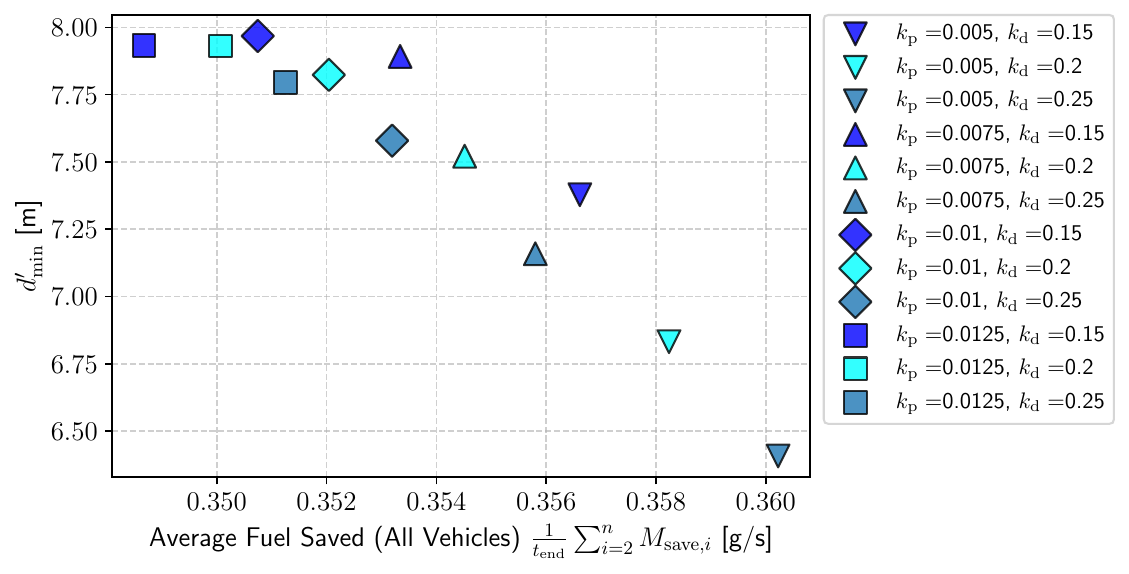}
\caption{Plot of 
minimum inter-vehicle distance $d'_{\mathrm{min}}$ versus average fuel saved for different control parameters
$k_{\mathrm{d}}\in\{0.15,0.2,0.25\}$,
$k_{\mathrm{p}}\in\{0.005,0.0075,0.01,0.0125\}$ under Constant Time Gap policy (CTG).}
\label{SummaryFuelvsSafety_CTGPF}
\end{figure}

\begin{figure}[t]
\centering 
\includegraphics[width=1\columnwidth]{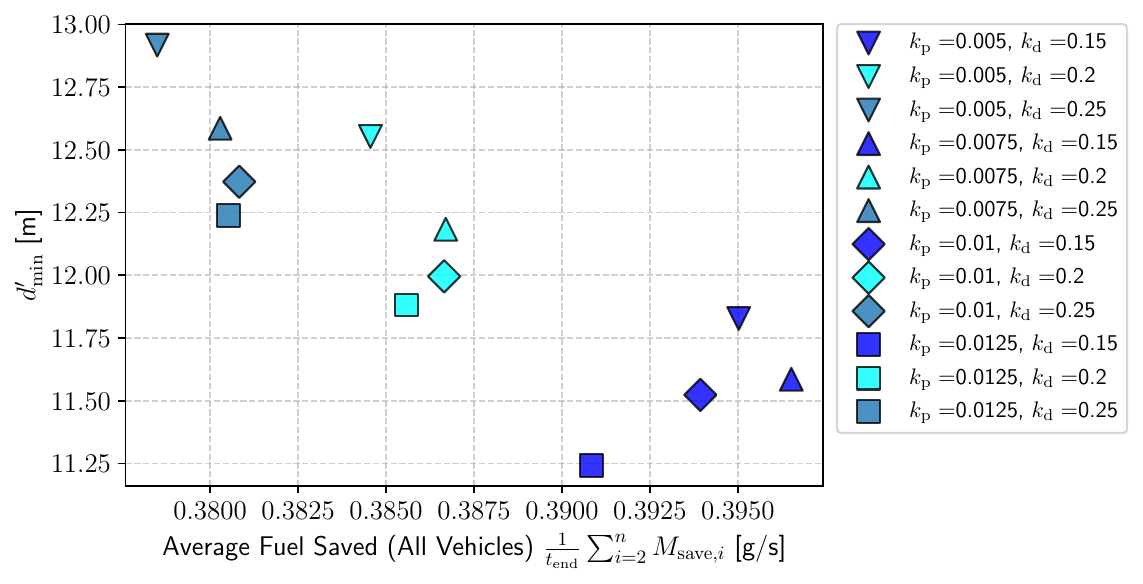}
\caption{Plot of 
minimum inter-vehicle distance $d'_{\mathrm{min}}$ versus average fuel saved for different control parameters
$k_{\mathrm{d}}\in\{0.15,0.2,0.25\}$,
$k_{\mathrm{p}}\in\{0.005,0.0075,0.01,0.0125\}$ under Constant Distance Gap policy (CDG).}
\label{SummaryFuelvsSafety_CDGPF}
\end{figure}

\section{Conclusion} \label{sec:Conclusion}
We investigated safety and fuel efficiency in vehicle platoons through a new simulation framework. Our proposed simulation framework allows reliable computation of the inter-vehicle distances across time and provides guaranteed bounds on the error induced by simulation. Using this framework, we first explored safety performance of different control parameters under sudden braking and packet losses causing transient communication outages. We then used our simulation framework for fuel efficiency analysis. 

In our safety simulations, we observed that random packet losses can be critical in that packet loss probability beyond a threshold can cause a sharp drop in safety. Regarding fuel efficiency simulations, we observed a safety and fuel-efficiency tradeoff driven by the control parameter $k_\mathrm{p}$. When the other control parameter ($k_\mathrm{d}$) is fixed, decreasing $k_\mathrm{p}$ increases the amount of saved fuel, but can cause a decrease in minimum inter-vehicle distance.

For future research we also aim to incorporate our time-discretized simulation software in traffic and vehicle simulation environments such as SUMO (Simulation of Urban MObility) platform \cite{sumo}, CARLA \cite{carla}, and NVIDIA Drive Sim \cite{nvidia_drive}. Furthermore, we plan to integrate OMNeT++~\cite{varga2008} as a method to achieve physically accurate characterizations of wireless transmission between vehicles.


\section*{Appendix} \label{sec:Appendix}
\begin{proof} Since $d_{i}(t)=q_{i}^{\top}x(t)-L_{i},$ by Lemma~\ref{Lemma-key}
and (\ref{eq:dt}), we have 
$d_{i}(t)-d_{i}(t_{k})  =q_{i}^{\top}(e^{A_{{\mathrm{c}}}(t-t_{k})}-I_{{\mathrm{c}}})x(t_{k})
  +q_{i}^{\top}\int_{0}^{t-t_{k}}e^{A_{{\mathrm{c}}}(t-t_{k}-\tau)}{\mathrm{d}}\tau B_{{\mathrm{c}}}u(t_{k})$, 
 where $I_{{\mathrm{c}}}\in{\mathbb{R}}^{(3+6n)\times(3+6n)}$ denotes
the identity matrix. Thus, by using Cauchy-Schwarz inequality
\begin{align*}
|d_{i}(t)-d_{i}(t_{k})| & =|q_{i}^{\top}(e^{A_{{\mathrm{c}}}(t-t_{k})}-I_{{\mathrm{c}}})x(t_{k})\\
 & \quad+q_{i}^{\top}\int_{0}^{t-t_{k}}e^{A_{{\mathrm{c}}}(t-t_{k}-\tau)}{\mathrm{d}}\tau B_{{\mathrm{c}}}u(t_{k})|\\
 & \leq\|q_{i}\|\|(e^{A_{{\mathrm{c}}}(t-t_{k})}-I_{{\mathrm{c}}})x(t_{k})\\
 & \quad\quad+\int_{0}^{t-t_{k}}e^{A_{{\mathrm{c}}}(t-t_{k}-\tau)}{\mathrm{d}}\tau B_{{\mathrm{c}}}u(t_{k})\|,
\end{align*}
which implies by triangle inequality that 
\begin{align}
|d_{i}(t)-d_{i}(t_{k})|\ & \leq\|q_{i}\|\Big(\|e^{A_{{\mathrm{c}}}(t-t_{k})}-I_{{\mathrm{c}}}\|\|x(t_{k})\|\nonumber \\
 & \quad+\int_{0}^{t-t_{k}}\|e^{A_{{\mathrm{c}}}(t-t_{k}-\tau)}\|{\mathrm{d}}\tau\|B_{{\mathrm{c}}}u(t_{k})\|\Big).\label{eq:d-ineq-1}
\end{align}
Now, note that
\begin{align}
&\|e^{A_c(t-t_k)}-I_\mathrm{c}\| = \left\|\int_{0}^{t-t_k}A_\mathrm{c} e^{A_\mathrm{c}s}\right\|\mathrm{d}s \nonumber \\
&\quad \leq \int_{0}^{t-t_k}\|A_\mathrm{c}\|\|e^{A_\mathrm{c}s}\|\mathrm{d}s = \|A_\mathrm{c}\| \int_{0}^{t-t_k}\|e^{A_\mathrm{c}s}\|\mathrm{d}s \nonumber \\
&\quad \leq \|A_\mathrm{c}\| \int_{0}^{t-t_k}e^{\|A_{\mathrm{c}}\|s}\mathrm{d}s = e^{\|A_{\mathrm{c}}\|(t-t_k)}-1. 
\label{eq:norm-inequality-D}
\end{align}
Similarly, $\|e^{A_{{\mathrm{c}}}(t-t_{k}-\tau)}\|\leq e^{\|A_{{\mathrm{c}}}\|(t-t_{k}-\tau)}$.
Therefore, (\ref{eq:d-ineq-1}) implies
\begin{align}
 & |d_{i}(t)-d_{i}(t_{k})|\nonumber \\
 & \,\,\leq\|q_{i}\|\Big(\left(e^{\|A_{{\mathrm{c}}}\|(t-t_{k})}-1\right)\|x(t_{k})\|\nonumber \\
 & \,\,\quad\quad+\int_{0}^{t-t_{k}}e^{\|A_{{\mathrm{c}}}\|(t-t_{k}-\tau)}{\mathrm{d}}\tau\|B_{{\mathrm{c}}}u(t_{k})\|\Big)\nonumber \\
 & \,\,=\|q_{i}\|\Big(\left(e^{\|A_{{\mathrm{c}}}\|(t-t_{k})}-1\right)\|x(t_{k})\|\nonumber \\
 & \,\,\quad\quad+\frac{e^{\|A_{{\mathrm{c}}}\|(t-t_{k})}-1}{\|A_{{\mathrm{c}}}\|}\|B_{{\mathrm{c}}}u(t_{k})\|\Big)\nonumber \\
 & \,\,=\|q_{i}\|\left(\|x(t_{k})\|+\frac{\|B_{{\mathrm{c}}}u(t_{k})\|}{\|A_{{\mathrm{c}}}\|}\right)\left(e^{\|A_{{\mathrm{c}}}\|(t-t_{k})}-1\right).\label{eq:d-good-ineq}
\end{align}
Furthermore, since for $t\in[t_{k},t_{k+1})$, we have $(t-t_{k})\leq t_{k+1}-t_{k}$.
By (\ref{eq:tk-condition1}), we obtain 
\begin{align}
e^{\|A_{{\mathrm{c}}}\|(t-t_{k})} & \leq e^{\|A_{{\mathrm{c}}}\|(t_{k+1}-t_{k})}
  \leq e^{\|A_{{\mathrm{c}}}\|\ln\left(\frac{\alpha}{\sqrt{2}(\|x(t_{k})\|+\frac{\|B_{{\mathrm{c}}}u(t_{k})\|}{\|A_{{\mathrm{c}}}\|})}+1\right)/\|A_{{\mathrm{c}}}\|}\nonumber \\
 & =\frac{\alpha}{\sqrt{2}(\|x(t_{k})\|+\frac{\|B_{{\mathrm{c}}}u(t_{k})\|}{\|A_{{\mathrm{c}}}\|})}+1\nonumber \\
 & =\frac{\alpha}{\|q_{i}\|(\|x(t_{k})\|+\frac{\|B_{{\mathrm{c}}}u(t_{k})\|}{\|A_{{\mathrm{c}}}\|})}+1,\label{eq:emu-ineq}
\end{align}
where in the last equality we used $\|q_{i}\|=\sqrt{2}$ that follows from
(\ref{eq:vij-def}). It then follows from (\ref{eq:d-good-ineq})
and (\ref{eq:emu-ineq}) that 
\begin{align*}
|d_{i}(t)-d_{i}(t_{k})| & \leq\|q_{i}\|\frac{\alpha}{\|q_{i}\|}=\alpha,
\end{align*}
which completes the proof. \end{proof}

\bibliographystyle{elsarticle-num}
\bibliography{references}

\end{document}